\documentclass[12pt]{amsart}
\numberwithin{equation}{section}
\usepackage{amsmath}
\usepackage{amsthm}
\usepackage{bm}
\usepackage{amssymb}
\usepackage{stmaryrd}
\usepackage{geometry}
\newcommand{\nicefrac}[2]{\ ^{#1}/_{#2}}

\usepackage{mathabx}
\usepackage{wasysym}

\usepackage{color}

\usepackage[normalem]{ulem}

\newtheorem{theorem}{Theorem}[section]
\newtheorem{lemma}[theorem]{Lemma}
\newtheorem{proposition}[theorem]{Proposition}
\newtheorem{corollary}[theorem]{Corollary}

\newtheorem{remark}[theorem]{Remark}

\newtheorem*{lt}{Theorem L($k$)}
\newtheorem*{lt0}{Theorem L($0$)}
\newtheorem*{ct}{Condition C($k$)}

\newcommand{\bt}{\begin{theorem}}
\newcommand{\et}{\end{theorem}}
\newcommand{\bl}{\begin{lemma}}
\newcommand{\el}{\end{lemma}}
\newcommand{\bc}{\begin{corollary}}
\newcommand{\ec}{\end{corollary}}
\newcommand{\bp}{\begin{proposition}}
\newcommand{\ep}{\end{proposition}}
\newcommand{\bpf}{\begin{proof}}
\newcommand{\epf}{\end{proof}}
\newcommand{\be}{\begin{equation}} 
\newcommand{\ee}{\end{equation}}
\newcommand{\bmul}{\begin{multline}}
\newcommand{\emul}{\end{multline}}
\newcommand{\bal}{\begin{align}} 
\newcommand{\eal}{\end{align}}
\newcommand{\beq}{\begin{eqnarray}}
\newcommand{\eeq}{\end{eqnarray}}
\newcommand{\ba}{\begin{array}}
\newcommand{\ea}{\end{array}}
\newcommand{\bma}{\begin{bmatrix}}
\newcommand{\ema}{\end{bmatrix}}
\newcommand{\bi}{\begin{itemize}}
\newcommand{\ei}{\end{itemize}}

\newcommand{\comm}[1]{}

\makeatletter
\newcommand{\opnorm}{\@ifstar\@opnorms\@opnorm}
\newcommand{\@opnorms}[1]{%
  \left|\mkern-1.5mu\left|\mkern-1.5mu\left|
   #1
  \right|\mkern-1.5mu\right|\mkern-1.5mu\right|
}
\newcommand{\@opnorm}[2][]{%
  \mathopen{#1|\mkern-1.5mu#1|\mkern-1.5mu#1|}
  #2
  \mathclose{#1|\mkern-1.5mu#1|\mkern-1.5mu#1|}
}
\makeatother
\newcommand{\twoline}[2]{\genfrac{}{}{0pt}{}{#1}{#2}}

\newcommand{\e}{\rm e}

\newcommand{\cB}{{\mathcal B}}

\newcommand{\cE}{{\mathcal E}}
\newcommand{\cF}{{\mathcal F}}
\newcommand{\cG}{{\mathcal G}}
\newcommand{\cH}{{\mathcal H}}

\newcommand{\cR}{{\mathcal R}}
\newcommand{\cM}{{\mathcal M}}

\newcommand{\cS}{{\mathcal S}}
\newcommand{\cT}{{\mathcal T}}

\newcommand{\bbC}{{\mathbb C}}

\newcommand{\bbE}{{\mathbb E}}

\newcommand{\bbN}{{\mathbb N}}

\newcommand{\bbR}{{\mathbb R}}

\newcommand{\bbZ}{{\mathbb Z}}

\newcommand{\bI}{{\bf I}}
\newcommand{\bJ}{{\bf J}}

\newcommand{\bS}{{\bf S}}

\newcommand{\fh}{{\mathfrak h}}

\newcommand{\dist}{\textrm{dist\ }} 
 
\newcommand{\wt}{\widetilde}  
  
\newcommand{\ol}{\overline}

\newcommand{\bks}{\backslash}

\newcommand{\ipc}[2]{\left \langle #1 , \ #2 \right \rangle }
\newcommand{\Ev}[1]{\E \left( #1 \right)}  
\newcommand{\Evc}[2]{\E \left ( #1 \middle | #2 \right )}
\newcommand{\norm}[1]{\left\Vert#1\right\Vert}
\newcommand{\abs}[1]{\left\vert#1\right\vert}
\newcommand{\set}[1]{\left\{#1\right\}}
\newcommand{\setb}[2]{\left \{ #1 \ \middle | \ #2 \right \} }
\newcommand{\com}[2]{\left[ #1 , #2 \right ]}
\newcommand{\bra}[1]{\left < #1 \right |}
\newcommand{\ket}[1]{\left | #1 \right >}
\newcommand{\dirac}[3]{\bra{#1} #2 \ket{#3}} 
\newcommand{\diracip}[2]{\left <#1 \middle | #2 \right >}

\newcommand{\bb}[1]{\mathbb{#1}}
\newcommand{\mc}[1]{\mathcal{#1}}
\renewcommand{\vec}[1]{\bm{#1}}

\def\Z{\mathbb Z}
\def\N{\mathbb N}
\def\R{\mathbb R}
\def\E{\mathbb E}

\def\e{\mathrm e}
\def\im{\mathrm i}
\def\di{\mathrm d}

\def\1{{\mathsf 1}}
\def\tem{\textemdash}

\def\Pr{\operatorname{Prob}} 
\def\diam{\operatorname{diam}} 
\def\ran{\operatorname{ran}} 
\def\dist{\operatorname{dist}}   
\def\dim{\operatorname{dim}}  
\def\esssup{\operatorname*{ess-sup}} 
\def\Im{\operatorname{Im}} 

\def\ra{\rightarrow}

\newcommand{\vm}{\vec{m}}

\newcommand{\vxi}{\vec{\xi}}

\newcommand{\vla}{\vec{\lambda}}

\newcommand{\vo}{\vec{0}}

\begin{document}
\pagestyle{plain}
\title{Localization in the Disordered Holstein model}
\author{Rajinder Mavi} 
\author{Jeffrey Schenker}
\address[R. Mavi and J. Schenker]{Michigan State University \\ Department of Mathematics \\ 619 Red Cedar Rd. \\ East Lansing, MI 48824}
\email[R. Mavi]{mavi.maths@gmail.com}
\email[J. Schenker]{jeffrey@math.msu.edu}
\thanks{R. M. was supported by a postdoctoral fellowship from the Michigan State University Institute for Theoretical and Mathematical Physics}\thanks{J.S. was supported by the National Science Foundation under Grant No. 1500386.}
\date{September 19, 2017}

\begin{abstract}
The Holstein model describes the motion of a tight-binding tracer particle interacting with a field of quantum harmonic oscillators.  We consider this model with an on-site random potential. Provided the hopping amplitude for the particle is small, we prove localization for matrix elements of the resolvent, in particle position and in the field Fock space. These bounds imply a form of dynamical localization for the particle position that leaves open the possibility of resonant tunneling in Fock space between equivalent field configurations.  
\end{abstract}

\maketitle
  \section{Introduction}

In this paper we study the phenomenon of  localization for a single particle in a disordered environment interacting with a field of quantum oscillators.   The main result derived here is that a version of Anderson localization, suitably defined, holds for low lying energy states of the model.  More precisely, for a given energy $E$, if the particle hopping amplitude is less than an energy dependent positive threshold, then the Green's function of the system Hamiltonian at an energy below $E$ decays exponentially off the diagonal with respect to particle position \emph{and} in a particular basis of the quantum oscillator Fock space. As a consequence of this estimate, we obtain a dynamical localization bound for the \emph{particle} position.  However, we do not obtain, and do not expect, complete localization of the oscillator excitations.

Localization was initially proposed \cite{Anderson1958} and studied in single particle models \tem see \cite{Lagendijk2009} for a summary of the physics literature.   There is a large body of rigorous results on single particle localization in the mathematics literature \tem the review \cite{disertori2008random} and recent book \cite{aizenman2015random} provide overviews of complementary approaches to the subject. In recent years, the proposed phenomenon of \emph{many body localization}  has been vigorously discussed in the physics literature \tem see for example \cite{Basko2006,Pal2010,Nandkishore2015}.  Despite the intense interest, there have been only a few attempts at rigorous results on the subject. Notable efforts are Imbrie's result on disordered spin chains \cite{Imbrie2016}, a recent paper \cite{Mastropietro2017} of Mastropietro on a quasi-periodic model, and a number of papers that use single particle results to derive consequences for integrable models, see \cite{Abdul-Rahman2017} and references therein.  There is an older mathematical literature on the $n$-particle regime \cite{AW,Chulaevsky2009}, which focused on localization in systems of finitely many particles in an arbitrarily large volume.  More recently, localization of the many body ``droplet spectrum'' of the attractive XXZ spin chain, and related models, was proved independently by two groups \cite{Beaud2017,Elgart2017}. These results present localization in a genuine many body context, although only in one dimension and in a regime of zero energy density.   The present work demonstrates localization in a many body system without particle number conservation, in arbitrary dimension, but also in the regime of zero energy density.  Our main interest is in obtaining localization bounds that extend to Fock space, a proposed characteristic of many body localization \cite{Basko2006} .

We  study a particle confined to a given region $\Lambda\subset\mathbb{Z}^D$. In formulating certain arguments, we may restrict the set $\Lambda$ to be finite. However  \emph{the estimates obtained do not depend on the size of  $\Lambda$}, and when suitably formulated extend \emph{a posteriori} to infinite volumes.
Let $\fh_\Lambda = \ell^2(\Lambda)$ denote the Hilbert space for a particle in $\Lambda$.  
The oscillator Hilbert space is  the Bosonic Fock space over $\ell^2(\Lambda)$, denoted here $\cF_\Lambda$.  The Hilbert space for the combined particle/oscillator system is $\cH_\Lambda := \fh_\Lambda \otimes \cF_\Lambda$, which we represent as \be\cH_{\Lambda} \ = \ \setb{\psi:\Lambda \rightarrow \cF_\Lambda}{\sum_{x\in \Lambda}\norm{\psi(x)}^2 < \infty} \ = \ \ell^2(\Lambda;\cF_\Lambda).  \ee

We  study the \emph{disordered Holstein Hamiltonian} on $\cH_{\Lambda}$:
\be H_\Lambda(\gamma) \ := \ H_{\Lambda}^{(\mathrm{Hol})} + V_\Lambda 
  	\ee
where $V_\Lambda$ is an onsite random potential,
\begin{equation}
V_\Lambda \psi(x) = v_x \psi(x)	,
\end{equation}
and $H_{\Lambda}^{(\mathrm{Hol})}$ is the \emph{Holstein Hamiltonian}, which describes a tight-binding particle interacting with independent harmonic oscillators at each site of the lattice. Here
\begin{enumerate}
\item $\{v_x\}_{x\in \Lambda}$ is a family of independent, identically distributed random variables, whose common distribution $\rho(v)\di v$ has a bounded density $\rho$ supported on $I_V = [0,V_+]$,	and
\item the Holstein Hamiltonian is the following operator on $\cH_\Lambda$:
\be H_{\Lambda}^{(\mathrm{Hol})} \ = \ \gamma \Delta_\Lambda-  (\alpha b_{\vec{X}}^\dagger + \alpha^* b_{\vec{X}}) + \omega \sum_{x\in \Lambda} b_x^\dagger b_x +\frac{|\alpha|^2}{\omega} ,\ee
where   $\gamma$, $\omega$ are non-negative parameters, $\alpha \in \bbC$, and
\begin{enumerate}
	\item $\Delta_\Lambda$ denotes the discrete Laplacian \be 
	\Delta_\Lambda\psi(x) = \sum_{\substack{x\sim y \\ y\in \Lambda}} \psi(x)-\psi(y),\ee
		where $x\sim y$ indicates that $\|x -y\| = 1$;
	\item $b_x$ is the annihilation operator for an oscillator excitation at $x$, that is the operators $\set{b_x,b_x^\dagger}_{x\in \Lambda}$ satisfy the canonical commutation relations,
	    \be  \com{b_x}{b_y}=0, \quad \com{b_x}{b_y^\dagger} = \delta_{x,y}\ ; \ee
	    and
	\item the operator $b_{\mathbf{X}}$ denotes the annihilation operator \emph{at the particle position}:
	   \be  b_{\mathbf{X}}\psi(x) = b_x \psi(x). \ee
\end{enumerate}
\end{enumerate}

One may expect, by analogy with the Anderson model, that disorder will inhibit propagation of the particle.   However, this does not follow directly from results on the Anderson model, because the particle may exchange energy with the oscillator system.  Nonetheless, this expectation is borne out by the results on dynamical localization derived below.  Before proceeding, we formulate a simple version of these results which may be stated without introducing various technical definitions. 

\bt \label{thm:dl} Fix $E>0$ and let $P_{[0,E)}(H_\Lambda(\gamma))$ denote the spectral projection for $H_\Lambda(\gamma)$ onto states with energies below $E$.  Let $\Phi_{0;\Lambda}\in \cF_\Lambda$ denote the vacuum state, i.e., 
\be\label{eq:Phi0form} \sum_{x\in \Lambda } b_x^\dagger b_x \Phi_{0;\Lambda} = 0.\ee
For any $\mu >0$, there is $\gamma_0(E,\mu)$ such that if $\gamma <\gamma_0(E,\mu)$ then there is $C<\infty$ so that
\be  \Ev{\sup_{t\in \bbR} \abs{\ipc{\delta_x\otimes \Phi_{0;\Lambda}}{ e^{-\im t H_\Lambda(\gamma)}P_{[0,E)}(H_\Lambda(\gamma))  \delta_y \otimes \Phi_{0;\Lambda}}}} \ < \ C\e^{-\mu \|x-y\|} , \ee
for any $\Lambda \subset \bbZ^d$ and any $x,y\in \Lambda$. 
Here $\|x-y\|$ denotes the graph distance between $x$ and $y$ and $\bbE$ denotes expectation with respect to the randomness. \et

Thm.\ \ref{thm:dl} is an easy consequence of Thm.\ \ref{thm:DL} below.   Our main result is Theorem \ref{thm:main}, which is an exponential clustering bound on the matrix elements of $\left (H_\Lambda(\gamma)-z \right )^{-1}$ in the orthonormal basis of eigenfunctions for $H_\Lambda(0)$, with $\gamma=0$.  To formulate the result, 
it is useful to decompose $H_\Lambda(\gamma)$ as follows:
\be 
     H_\Lambda(\gamma) \
       = \ \gamma \Delta_\Lambda+ H_\Lambda(0)
   \ee
where
\begin{equation} H_\Lambda(0)\ =\ \omega H_{\Lambda}^{(\mathrm{ph})}+ V_\Lambda , \label{eq:H0}\end{equation} 
and
\be\label{eq:Hph}
H_{\Lambda}^{(\mathrm{ph})} \ = \ \sum_{x\in \Lambda}b_x^\dagger b_x - \left ( \frac{\alpha}{\omega} b_{\vec{X}}^\dagger + \frac{\alpha^*}{\omega} b_{\vec{X}} \right ) + \frac{ |\alpha|^2}{\omega^2}.\ee
To state Theorem \ref{thm:main}, we must first describe the  explicit orthonormal eigenbasis for $H_\Lambda(0)$,  indexed by particle position and oscillator excitation numbers.

\subsection{Description of $H_\Lambda(\gamma)$ in the eigenbasis for $H_\Lambda(0)$}  The two operators which sum in eq.\ \eqref{eq:H0} to define $H_\Lambda(0)$ commute with one another.  The ``phonon'' operator $H_\Lambda^{(\mathrm{ph})}$ can be written as follows:
\begin{equation}
H_{\Lambda}^{(\mathrm{ph})} \ = \ \sum_{x\in \Lambda} a_{x}^\dagger a_x
\end{equation}
with 
\begin{equation}
a_x = b_x -\beta I[\mathbf{X}=x], \quad \beta = \frac{\alpha}{\omega}.
\end{equation}
Note that the family $\{a_x^\dagger,a_x\}_{x\in \Lambda}$ satisfies the canonical commutation relations,  however these operators do not commute with the particle kinetic energy $\Delta_\Lambda$.

A single harmonic oscillator has a Hilbert space spanned by states $\{\ket{m}\}_{m=0}^\infty$ with Hamiltonian  $b^\dagger b \ket{m}=m\ket{m}$, $\com{b^\dagger}{b}=1.$ Recall that
 \be  (b^\dagger - \beta^*)(b  -  \beta) D(\beta) | m \rangle
            = D(\beta) (b^\dagger b ) | m \rangle = m D(\beta) |m\rangle \ee
where $D(\beta)$ is the unitary \emph{Glauber displacement operator} \cite{Glauber1963},
      \be
      D(\beta) = \e^{\beta b^\dagger - \beta^* b} 
              = \e^{-\frac{1}{2}|\beta|^2}  \e^{\beta b^\dagger} \e^{ - \beta^* b }. \ee
Since $D(\beta)$ is unitary,  $\sum_{\xi} \abs{\dirac{\xi}{D(\beta)}{m}}^2 = 1 .$ We derive in Appendix \ref{displacement}
the following  bound on  matrix elements of $D(\beta)$. 
 \bp     \label{disp0}
    Let $\beta \in \bb{C}$.
    For any $\mu>0$ and $0 < p \le  2$ there is $C_{\mu,\beta,p}<\infty$ such that
\be\label{dispsum}  \sum_{m=0}^\infty \e^{\mu|\sqrt{n} - \sqrt{m}| } \abs{\dirac{m}{D(\beta)}{n}}^p < C_{\mu,\beta,p}  \left ( n \vee 1 \right )^{\frac{1}{2}-\frac{p}{4}} \ .
\ee
In particular,  \begin{equation} \label{dispdecay}
	\abs{\dirac{m}{D(\beta)}{n}} \ \leq \  A_{\mu,\beta}  \,  \e^{-\mu|\sqrt n- \sqrt m|} ,
\end{equation} 
with $A_{\mu,\beta} < \infty$.
 \ep 

\begin{remark}
	Here and throughout  we use the notation $n \vee m = \max(n,m)$.  Similarly, $n\wedge m=\min(n,m)$.
\end{remark}

We can now write down an eigenbasis for $H_\Lambda(0)$. Let 
\be\label{eq:MLambda} \cM_\Lambda \ := \  \setb{\vec{m}}{\vec{m}:\Lambda \ra \{0,1,2,\ldots\} \quad \text{and} \sum_{x\in \Lambda} \vm(x) <\infty}.\ee
The condition $\sum_x \vm(x)<\infty$ is equivalent to the condition ``$\vm(x)=0$ for all but finitely many $x$,'' which is automatic for $\Lambda$ finite.
Given $\vec{m}\in \cM_\Lambda$ let
    \be  \ket{\vec{m}} \ := \ \prod_{x \ :\ \vec{m}(x)\neq 0} \frac{1}{\sqrt{\vec{m}(x)!}} \left (b^\dagger_x \right )^{\vec{m}(x)}  \ket{\vec{0}},\ee
where $\ket{\vec{0}}$ is the vacuum vector in $\cF_\Lambda$.  The set $\setb{ \ket{\vec{m}}}{\vec{m}\in \cM_\Lambda}$
     is  basis for the oscillator Fock space $\cF_{\Lambda}$.  Given $x\in \Lambda$ and $\vec{m}:\Lambda \rightarrow \bbN$, let
      \be D_x(\beta) \ket{\vec{m}} \ := \ \e^{\beta b_x^\dagger-\beta^*b_x} \ket{\vec{m}}\ee
and 
\be \label{eq:ketxm} \ket{x,\vec{m}}  \ := \ \ket{x} \otimes D_x(\beta) \ket{\vec{m}}.\ee
Thus for $\vec{m} \neq \vec{0}$, we have 
\be  \ket{x,\vec{m}} \ =  \ 
\prod_{u\in \Lambda} \frac{(a^\dagger_u)^{\vec{m}(u)}}{\sqrt{\vec{m}(u) !}} \ket{x,\vec{0}} ,
     \ee
     where the product is, in fact, finite since all but finitely many terms are the identity.

The vectors defined in eq.\ \eqref{eq:ketxm} are eigenstates of $H_\Lambda(0)$, with corresponding eigenvalues
\be  H_\Lambda(0)  \ket{x,\vec{m}} = (  \omega N_\Lambda(\vec{m})   + v_x  )\ket{x,\vec{m}}    \ee
     where 
\be \label{eq:NLambda} N_\Lambda(\vec{m})  = \sum_{x\in \Lambda} \vec{m}(x). 
\ee
Furthermore the collection of eigenstates
     \be \label{eq:cELambda}
     \cE_\Lambda = 
     \setb{\ket{x,\vec{m}} }{x\in \Lambda,\ \vec{m}\in \cM_\Lambda }\ee
     is an orthonormal basis for $\cH_\Lambda$. Thus the spectrum of $H_\Lambda(0)$ is 
     $$\sigma(H_\Lambda(0)) = \overline{\setb{ v_x + \omega k}{ x\in \Lambda \ \text{and} \ k\in \bbN }} \ \subset \ \bigcup_{k = 0}^\infty  [\omega k,\omega k + V_+].  $$
For finite $\Lambda$, the spectrum is a random, countable set, consisting of $|\Lambda|^k$ points in each interval $[\omega k,\omega k+V_+]$. For infinite $\Lambda$, it is equal to $\bigcup_{k=0}^\infty [\omega k + \overline{\esssup \rho}],$
with probability one.

The kinetic energy operator $\Delta_\Lambda$ is not diagonal in the basis $\cE_\Lambda$, but has the explicit matrix elements
 \be \label{Kdef}
 \dirac{x,\vec{m}}{\Delta_\Lambda}{y,\bm{\xi}} \ = \ K_{x,\bm{m}}^{y,\bm{\xi}} 1_{x,y\in \Lambda} ,
 \ee
 where
 \begin{align}
 	 K_{x,\bm{m}}^{y,\bm{\xi}} \ 
       :=& \ 2D 1_{x=y} \diracip{\bm{m}}{\bm{\xi}} - 1_{x\sim y} \dirac{\vec{m}}{ D_x(\beta) D_y(\beta)}{\vec{ \xi}} \\ \nonumber
 	 =& \ \begin{cases} 2 D & \text{ if } x=y \text{ and } \bm{m}=\bm{\xi} , \\
 	  \langle \bm{m}(x) | D(-\beta) |\bm{\xi}(x)\rangle 
                    \langle \bm{m}(y) | D(\beta) | \bm{\xi}(y)  \rangle & \text{ if } x\sim y \text{ and } \bm{m}(u)=\bm{\xi}(u) \text{ for }u\neq x,y , \\
                    0 & \text{ otherwise.}
 \end{cases}
 \end{align}
As $D(\beta)$ is a unitary operator, we have  $\sum_{\bm{m}}  |K_{x,\bm{m}}^{y,\bm{\xi}}|^2 = 2D 1_{x=y} - 1_{\|x-y\| = 1} $.  In particular, $0\le \Delta_\Lambda\le 4D$, which also follows from the fact that  $\Delta_\Lambda$ is the same as $\wt{\Delta}_\Lambda \otimes 1$ acting on the tensor product $\fh_\Lambda \otimes \cF_\Lambda $ where $\wt{\Delta}_\Lambda$ is the usual discrete Laplacian on $\fh_\Lambda$. Let 
\begin{equation}\label{eq:bands} \bI_k \ = \   [\omega k   , \omega k + V_+ +    4D \gamma    ].\end{equation}
Then the spectrum of $H_\Lambda(\gamma)$ is contained in $\Sigma = \bigcup_{k = 0}^\infty \bI_k  $.

\subsection{Localization} We characterize localization in terms of  the amplitude of the \emph{Green's function}, 
\be \label{eq:GLambda} G_\Lambda (x,\bm{m};y,\bm{\xi} ;z)   \ := \ \dirac{x,\bm{m} }{ (H_\Lambda(\gamma) - z)^{-1}}{ y,\bm{\xi}},  \ee 
for $\ket{x,\bm{m}},\ket{y,\bm{\xi}} \in \cE_\Lambda $.  Localization is signified by exponential decay of $G_\Lambda(x,\bm{m};y,\bm{\xi};z)$ in a suitable metric on $\cE_\Lambda$, which we now specify. 

Let $\|x - y\|$ denote the graph metric on $\Lambda$, i.e.,
\be \label{eq:metric} \|x-y\| \ := \ \text{length of the shortest lattice path from $x$ to $y$} \ = \ \sum_{i=1}^{D} |x_i-y_i|. \ee
Let $R_{\bm{m}|\bm{\xi}}(x)$ be the radius of the smallest ball centered at $x$ containing all sites where $\bm{m}>0$ and not-equal to $\bm{\xi}$, i.e., 
\be \label{eq:Rmxix}   R_{\bm{m} | \bm{\xi}} (x)  \ := \  \max \setb{\|u- x\|  }{ \bm{m}(u) >0 ; \bm{m} (u)\neq \bm{\xi} (u)   }  ,
  \ee
  with the convention that the maximum over the empty set is $0$.
We define the following function on $\cE_\Lambda^2$,
    \be  \label{ups}
      \Upsilon_\Lambda(x,\vm; y,\vxi) \ := \    \max \set{ \|x - y\|, \ R_{\vm |  \vxi} (x) , \ R_{\vxi| \vm}(y) } .  \ee
In Prop.\ \ref{metric} below, we show that  $\Upsilon_\Lambda$  is a pseudo-metric\footnote{That is $\Upsilon_\Lambda$ is non-negative, vanishes on the diagonal, is symmetric and satisfies the triangle inequality.  However, $\Upsilon_\Lambda$ may vanish for $\ket{x,\vm} \neq \ket{y,\vxi}$.} on $\cE_\Lambda$ and that $\Upsilon_\Lambda(x,\vm; y,\vxi)=0$ if and only if $x=y$ and $\vm(x')=\vxi(x')$ for $x'\neq x$. 

Our bound on the Green's function  also depends on a family of pseudo-metrics on $\cM_\Lambda$ that further control the distance between oscillator states in Fock space.  Motivated by Prop.\ \ref{disp0}  define  the following \emph{metric}:
\be \label{eq:metric}
r_\Lambda(\vec{m},\vec{\xi}) \ := \ \sum_{x\in \Lambda} \sqrt{\abs{\vm(x)-\vxi(x)}}.\ee
This metric will appear in the Combes-Thomas bound, Thm.\ \ref{thm:specCT} below, but is stronger than the pseudo-metric that appears in our localization bound.
For $k \geq 0$,
let $\cM_{\Lambda}^{(k)}$ denote the set of  oscillator configurations with total excitation number $N_\Lambda(\vm)=k$, i.e.,
\be\label{eq:cMk} \cM_{\Lambda}^{(k)}  \ := \ \setb{ \vm \in \cM_\Lambda}{ N_\Lambda(\vm) = k }. 
\ee
Thus, $\cM_\Lambda^{(k)}$ consists of those oscillator states that are ``on-shell'' for the  $k$-th band.   For each $k$, we let $\cR_\Lambda^{(k)}$ denote the  pseudo-metric  on $\cM_\Lambda$ obtained from $r_\Lambda$ by collapsing the elements of $\cM_\Lambda^{(k)}$ to a point:
 \be \label{eq:cRk}  \cR^{(k)}_\Lambda(\vm,\vxi)  
        =   \min\left\{ r_\Lambda(\vm,\vxi),
                 r_\Lambda(\vm,\cM_\Lambda^{(k)}) + r_\Lambda(\vxi,\cM_\Lambda^{(k)})  \right\} ,
         \ee
         where
         $ r_\Lambda(\vm,\cM_\Lambda^{(k)})  =   \min \setb{r_\Lambda(\vm,\vxi)}{\vxi\in \cM_\Lambda^{(k)}}.$
Note that $\mc{R}_\Lambda^{(k)}(\vm,\vxi)=0$ if both $\vm,\vxi\in \cM_\Lambda^{(k)}$.  Furthermore, for any band $k$
\be\label{eq:rlower} r_\Lambda(\vm,\vxi) \ \ge \ \cR_\Lambda^{(k)}(\vm,\vxi) \ \ge \ \sqrt{N_\Lambda(\vm)-N_\Lambda(\vxi)}.\ee

It follows, for each $k$, that $ \Upsilon_\Lambda(x,\vec{m};y,\vec{\xi}) + \cR^{(k)}_\Lambda(\vec{m},\vec{\xi})$ 
is a metric on $\cE_\Lambda$.\footnote{Since it is a sum of two pseudo-metrics, $ \Upsilon_\Lambda +\cR^{(k)}$ is itself a pseudo-metric.  Furthermore, if  $ \Upsilon_\Lambda(x,\vec{m};y,\vec{\xi}) + \cR^{(k)}_\Lambda(\vec{m},\vec{\xi})$ vanishes, then  $\Upsilon_\Lambda(x,\vec{m};y,\vec{\xi})=0$ and  $N_\Lambda(\vm)=N_\Lambda(\vxi)$.  From $\Upsilon_\Lambda(x,\vec{m};y,\vec{\xi})=0$ we conclude that $x=y$ and $\vm(x')=\vxi(x')$ for $x\neq x'$.  Since $N_\Lambda(\vm)=N_\Lambda(\vxi)$, we must have $\vm(x)=\vxi(x)$, and thus $\vm\equiv \vxi$.}   
Our main technical result is that for small enough $\gamma$  the Green's function  $G_\Lambda(x,\bm{m};y,\bm{\xi};E+\im \epsilon)$ for energies $E$ in the $k$-th band $\bI_k$ is exponentially small in  $\Upsilon_\Lambda(x,\vec{m};y,\vec{\xi}) + \cR^{(k)}_\Lambda(\vec{m},\vec{\xi})$.  We restrict the imaginary part of the energy, $\epsilon$, to be less than one, and introduce the notation 
\be \bS_k \ = \ \setb{E+\im \epsilon}{E\in \bI_k \text{ and } |\epsilon| <1}.\ee
\bt\label{thm:main} Suppose that $V_+ < \omega$ and fix $s<1$ and $\mu >0$. For each  $k=0,1,\ldots$ there is $\gamma_{k} >0$, depending on $k,\mu,s,V_+,\omega,\beta$, and $D$, such that if $\gamma <\gamma_{k}$, then there is $A <\infty$ so that for $z\in \bS_k$, 
     \be  
    \Ev{  |G_\Lambda(x,\vm;y,\vxi;z)|^s }   \leq         
                   A \e^{ - \mu \left(  \Upsilon_\Lambda(x,\vm;y,\vxi) + \cR_\Lambda^{(k)}(\vm,\vxi)\right ) }
   \ee 
for any $\Lambda \subset \bbZ^d$ and any  $\ket{x , \vm}$, $\ket{y,\vxi} \in \cE_\Lambda$.
\et

Taking $\gamma < \wt{\gamma}_k := \min_{0\le j\le k} \gamma_j$, we obtain localization throughout $\bS_0 \cup \cdots \cup \bS_k$. To state a bound that is uniform over all bands, we use the lower bound eq.\ \eqref{eq:rlower} for $\mc{R}_\Lambda^{(k)}$. Thus
\begin{corollary}\label{cor:main} Suppose that $V_+<\omega$ and fix $s <1$ and $\mu >0$.  For each $k=0,1,\ldots$ there is $\wt{\gamma}_k >0$ such that if $\gamma<\wt{\gamma}_k$ then there is $C <\infty$ so that for $E \le  \omega k + V_+ +4D\gamma$,
\be  
    \Ev{  |G_\Lambda(x,\vm;y,\vxi;E+\im \epsilon)|^s }   \leq         
                   C\e^{ - \mu \left(  \Upsilon_\Lambda(x,\vm;y,\vxi) + \abs{\sqrt{N_\Lambda(\vm)}-\sqrt{N_\Lambda(\vxi)}} \right ) }
   \ee 
for any $\Lambda \subset \bbZ^d$ and $|x , \vm\rangle,|y,\vxi \rangle \in \cE_\Lambda$.	
\end{corollary}

\subsection{Dynamical Localization} Adapting methods developed for continuum \cite{Aizenman2006} and multi-particle \cite{AW} Schr\"odinger operators, we may use decay of the Green's function to control the dynamical behavior of solutions to the time dependent Schr\"odinger equation.  Because Cor.\ \ref{cor:main} gives localization only in a region at the bottom of the spectrum, we will only prove dynamical localization for initial states that are localized in energy.  For each $k$, let 
\be \bJ_k \ = \ \setb{\lambda}{0\le \lambda  \le  \omega k+V_++4D\gamma} \ee
and let $\cB_1(\bJ_k)$ denote the set of all complex Borel measurable functions supported on $\bJ_k$ and point-wise bounded by one:
\begin{multline} \cB_1(\bJ_k) \ = \ \left \{ f:\bbR \rightarrow \bbC \quad \right | \quad f\text{ is Borel measurable},  \\ \left . |f(\lambda)|\le 1 \text{ and } f(\lambda)=0 \text{ for } \lambda \ge \omega k+V_++4D\gamma \right \}.\end{multline}
Then we have the following
\bt \label{thm:DL} Suppose that $V_+<\omega$ and fix $\mu >0$ and $\epsilon >0$.  For each $k=0,1,\ldots$ there is $\wt{\gamma}_{k}$ such that if $\gamma <\wt{\gamma}_{k}$ then there is $C_{k}$ so that
\be \label{eq:DL}\Ev{ \sup_{f \in \cB_1(\bJ_k)} \abs{\dirac{x,\vm}{f(H_\Lambda)}{y,\vxi}}} \ \le \ \frac{C_{k}}{1+N_\Lambda(\vxi)^{\frac{1}{2}-\epsilon}} \e^{-\mu\|x-y\|} \ee
for any $\Lambda \subset \bbZ^d$ and  $|x , \vm\rangle,|y,\vxi \rangle \in \cE_\Lambda$. 
 \et

\begin{remark} Since the left hand side is symmetric with respect to $\ket{x,\vm} \leftrightarrow \ket{y,\vxi}$, it follows from eq.\ \eqref{eq:DL} that
\be \label{eq:DL'}\Ev{ \sup_{f \in \cB_1(\bJ_k)} \abs{\dirac{x,\vm}{f(H_\Lambda)}{y,\vxi}}} \ \le \ \frac{C_{k}}{1+N_\Lambda(\vm)^{\frac{1}{4}-\frac{\epsilon}{2}}N_\Lambda(\vxi)^{\frac{1}{4}-\frac{\epsilon}{2}}} \e^{-\mu\|x-y\|} \ee
\end{remark}

Fixing $E>0$ and applying eq.\ \eqref{eq:DL} to the family of functions 
 \be f_t(\lambda) = \e^{-\im t \lambda } I[\lambda  < E]\ee
 results in the following bound
 \be \label{eq:trueDL}
 \Ev{ \sup_{t \in \bbR}  \abs{\dirac{x,\vm}{\e^{-\im t H_\Lambda } P_{[0,E)}(H_\Lambda)}{y,\vxi}}} \ \le \ C \e^{-\mu\|x-y\|},\ee
 provided $E<\omega(k+1)$ and $\gamma < \wt{\gamma}_k$.  Theorem \ref{thm:dl} follows immediately. 

The proof of Thm.\ \ref{thm:DL} is based on estimates for \emph{eigenfunction correlators} associated to $H_\Lambda$ for \emph{finite} $\Lambda \subset \bbZ^d$.  Given finite $\Lambda$, the operator $H_\Lambda$ has compact resolvent, and thus pure point spectrum \tem indeed, the restriction of $H_\Lambda$ to $\bJ_k$ is finite dimensional.  Following \cite{AW}, we  define the eigenfunction correlator on $\bJ_k$ as follows
\be Q_\Lambda(x,\vm;y,\vxi;\bJ_k) \ = \ \sum_{E\in \sigma(H_\Lambda)\cap \bJ_k}
 \abs{\dirac{x,\vm}{P_{\{E\}}(H_\Lambda)}{y,\vxi}},
\ee
where $P_{\{E\}}(H_\Lambda)$ denotes the spectral projection onto the eigenspace for $H_\Lambda$ corresponding to eigenvalue $E$. 
Note that for any bounded, Borel measurable function $f:\bbR\rightarrow \bbC$ that vanishes on $\bJ_k^c$, we have
\be  \abs{\dirac{x,\vm}{f(H_\Lambda) }{y,\vxi}} \ \le \ \left [\sup_{E\in \bJ_k} \abs{f(E)}  \right ] Q_\Lambda(x,\vm;y,\vxi;\bJ_k). \label{eq:f} 
\ee

The correlator $Q_\Lambda$ is defined only for finite $\Lambda$; however bounds on $Q_\Lambda$ that are uniform in $\Lambda$ can be used to control matrix elements of $f(H_\Omega)$ for infinite $\Omega$.  Theorem 4.1 of \cite{AW}, adapted to the present context, gives a result of this type.
\bt \label{thm:ec2dl} Let $\Omega\subset \bbZ^d$ and suppose that for some increasing sequence $\Lambda_n$, $n=1,2,\cdots$, of finite domains converging to $\Omega$ the following bound holds
\be \Ev{Q_{\Lambda_n}(x,\vm;y,\vxi;\bJ_k)} \ \le \ A \e^{-K(x,\vm;y,\vxi)}.\ee
Then
\be \Ev{ \sup_{f\in \cB_1(\bJ_k)} \abs{\dirac{x,\vm}{f(H_\Omega) }{y,\vxi}} } \ \le \ A \e^{-K(x,\vm;y,\vxi)}.\ee
\et

\begin{proof} Essentially this boils down to three facts:
\begin{enumerate}
\item As $n\rightarrow \infty$, $H_{\Lambda_n}$ converges to $H_\Omega$ in the strong resolvent sense.
\item Given  vectors $\phi,\psi$ in a Hilbert space and a self-adjoint operator $H$, let $\mu_{\phi,\psi;H}$ denote the spectral measure for $H$ associated to $\phi$ and $\psi$; i.e., $ \int_{\bbR} f(E) \mu_{\phi,\psi;H}(\di E) \ = \ \ipc{\phi}{f(H)\psi}.$
	Then for any interval $I\subset \bbR$, we have $\sup_{f\in \cB_1(I)} \ipc{\phi}{f(H)\psi}\ = \ |\mu_{\phi,\psi}| (I),$ where $|\mu_{\phi,\psi}|$ denotes the total variation of $\mu_{\phi,\psi}.$
\item If a sequence $H_n$  of self-adjoint operators converges to a limit $H$ in the strong resolvent sense, then for any fixed pair $\phi,\psi$, the sequence of spectral measures $\mu_{\phi,\psi;H_n}$  converges to $\mu_{\phi,\psi;H}$ in the weak-$\star$ topology.
\end{enumerate}
For further details, see the proof of Theorem 4.1 in \cite{AW}.
\end{proof}

Bounds on the correlator $Q_\Lambda$ in terms of the Green's function are provided by the following
\bt \label{thm:gf2ec} For each $s<1$ there is $C_s$, depending on $s$ and the distribution $\rho$ of the random potentials but not on the volume $\Lambda$ or the band index $k$, such that
\begin{multline}\label{eq:gf2ec} \Ev{Q_\Lambda(x,\vm;y,\vxi;\bJ_k)} \\ \le \ C_s \left [ \frac{k\vee 1}{N_\Lambda(\vxi)\vee 1} \right ]^{1-\frac{s}{2-s}} \left [  \sum_{\vla \in \cE_\Lambda} \int_{\bJ_k} \Ev{\abs{G_\Lambda(x,\vla;y,\vxi;E)}^s}  \di E \right ]^{\frac{1}{2-s}} .\end{multline}
\et
\noindent In Appendix \ref{sec:DL} we explain how to adapt the proof of Theorem 4.5 in \cite{AW} to obtain this result.  

The sum over $\vla$ on the right hand side of eq.\ \eqref{eq:gf2ec} is a consequence of an argument in the proof based on averaging over the potential $v_x$ at site $x$. Note that $H_\Lambda = \wt{H} + v_x P_x$ where $\wt{H}$ is independent of $v_x$ and $P_x$ is the projection onto the closed linear span of $\{\ket{x,\vla}\}_{\vla \in \cE_\Lambda}$.  Thus this sum amounts to a sum over a basis of states affected directly by the potential $v_x$. At first sight, the emergence of the sum on the right hand side of eq. \eqref{eq:gf2ec} might seem to be a technical issue.  However, there is a physical reason behind it.  Without the hopping $\gamma \Delta_\Lambda$, any two states $\ket{x,\vla_1}$ and $\ket{x,\vla_2}$ with $N_\Lambda(\vla_1)=N_\Lambda(\vla_2)$ are exactly resonant.  Thus the right hand side of eq.\ \eqref{eq:gf2ec} includes a sum over states nearly resonant to $\ket{x,\vm}$, to which we cannot rule out tunneling.  

In \cite{Mavi2017} we study resonant tunneling of this type in a simplified model.  In that paper, we study two copies of the Anderson model on $\bbZ^d$, with identical disorder, coupled by a single hopping term connecting the origins of the two copies of $\bbZ^d$.  The Hamiltonian is a random operator $H$ on the graph $\Gamma = \bbZ^d \times \{0,1\}$  with graph metric
\begin{equation}
d_\Gamma(x,i;y,j) \ = \ \begin{cases} \|x-y\| & \text{if } i=j, \text{and}\\ 1+
 \|x\|+\|y\| & \text{if } i\neq j.	
 \end{cases}
\end{equation}
The operator $H$ is of the form $-\gamma \Delta + V$ where $\Delta$ is the nearest neighbor graph Laplacian on $\Gamma$ and $V = \sum_{x} v_x \left (\sum_{i=0,1} \ket{x,i}\bra{x,i} \right )$.
In \cite{Mavi2017}, we show that, for small enough $\gamma$,
\begin{enumerate}
	\item the fractional moments of the matrix elements of the resolvent of $H$ decay exponentially \emph{with respect to the graph metric},
	\be \Ev{\abs{\dirac{x,i}{\frac{1}{H-z}}{y,j}}^s} \ \le \ A \e^{-\mu d_\Gamma(x,i;y,j)};\ee
	\item the matrix elements of the Schr\"odinger evolution are bounded as follows
	\be \Ev{\sup_{t\in \bbR} \abs{\dirac{x,i}{\e^{-\im tH} }{y,j}}} \ \le \ C \e^{-\mu \|x-y\|};\label{eq:cartoonDL} \ee
	\item but eq.\ \eqref{eq:cartoonDL} cannot be improved to give decay in the metric $d_\Gamma$, because for any given realization of the randomness there is a sequence of vectors $\phi_k$ localized around points $x_k \rightarrow \infty$ such that for each $k$ there is a time $t_k$ at which
	\be \abs{\dirac{\phi,0}{\e^{-\im t_k H}}{\phi,1}} \ > \ 1 - \e^{-\epsilon |x_k|}.\ee
\end{enumerate}
That is, at time $t_k$ there is nearly perfect tunneling from the state $\ket{\phi_k,1}$ to the state $\ket{\phi_k,0}$, despite the fact that both states are localized far from the bond at the origin connecting the two copies of $\bbZ^d$.  

Our control on the spectrum and dynamics of the random Holstein hamiltonians considered here is not precise enough to carry out an analysis as in \cite{Mavi2017}.  However, the results of that paper strongly suggest that we cannot expect to do better here.  
 
Putting together Cor.\ \ref{cor:main} and Thms.\ref{thm:ec2dl} and \ref{thm:gf2ec} we obtain the following
\begin{lemma}\label{lem:DL}	Suppose that $V_+<\omega$ and fix $\mu >0$ and $\epsilon >0$.  For each $k=0,1,\ldots$ there is $\wt{\gamma}_{k}$ such that if $\gamma <\wt{\gamma}_{k}$ then there is $C_k<\infty$ so that
\begin{multline} \label{eq:DL2}\Ev{ \sup_{f \in \cB_1(\bJ_k)} \abs{\dirac{x,\vm}{f(H_\Lambda)}{y,\vxi}}} \\ \le \ C_{k,\epsilon} \left [ \frac{1}{N_\Lambda(\vxi)\vee 1} \right ]^{1-\epsilon }  \sum_{\vla \in \cE_\Lambda} \e^{-\mu \left  (\Upsilon_\Lambda(x,\vla;y,\vxi) + \abs{\sqrt{N_\Lambda(\vla)}-\sqrt{N_\Lambda(\vxi)}} \right )} \end{multline}
for any $\Lambda \subset \bbZ^d$ and  $|x , \vm\rangle,|y,\vxi \rangle \in \cE_\Omega$. 
\end{lemma}

We close this section by explaining how Lem.\ \ref{lem:DL} implies Thm.\ \ref{thm:DL}.  By Lem.\ \ref{lem:expsumbound} below
\be \sum_{\vla \in \cE_\Lambda} \e^{-\mu \left  (\Upsilon_\Lambda(x,\vla;y,\vxi) + \abs{\sqrt{N_\Lambda(\vla)}-\sqrt{N_\Lambda(\vxi)}} \right )} \ \le \ C_{\mu,\nu} \sqrt{N_\Lambda(\vxi) \vee 1} \e^{-\nu \|x-y\|}\ee
for any $0 <\nu <\mu$. Thus eq.\ \eqref{eq:DL2} implies eq.\ \eqref{eq:DL} with a modified value of the decay constant $\mu$.  Thus Thm.\ \ref{thm:DL} follows from Lem.\ \ref{lem:DL}.

 \subsection{A sketch of the proof} The proof of Thm.\ \ref{thm:main} will be accomplished by induction on the band number $k$. We  refer to the specialization of Thm.\ \ref{thm:main} to the $k$-th band as:
\begin{lt} 
  For any $\mu > 0$ and $0 <s < 1 $ there is $\gamma_{k,\mu,s}$  so that if $0 < \gamma < \gamma_{k,\mu,s}$, then there is $A_{k,\mu,s} <\infty$ such that 
   \be  \Ev{\abs{G_\Lambda(x,\vm;y,\vxi,z)}^s} \ \le \ A_{k,\mu,s} \e^{-\mu \Upsilon_\Lambda (z,\vm;y,\vxi) -\mu \cR^{(k)}_\Lambda(\vm,\vxi)} \ee 
   for any  $\Lambda \subset \Z^d$,  $|x,\vm\rangle, |y,\vxi\rangle \in  \cE_{\Lambda}$,  and   $z\in \bS_k$.
\end{lt}

The general strategy of the proof is to show that Thms.\ L($j$) for $j=0,\ldots,k-1$ together imply Thm.\ L($k$). For each $k$, the basic tool will be an inequality of the form
\begin{equation}\label{eq:basicstep}
	\Ev{\abs{G_\Lambda(x,\vm;y,\vxi,z)}^s} \ \le \ \sum_{u\not \in \Gamma,\vla } K_\Lambda (x,\vm;u,\vla) \Ev{\abs{G_{\Lambda\setminus \Gamma} (u,\vla;y,\vxi,z)}^s},
\end{equation}
valid when $\|x-y\|$ is not too small.  Here $\Gamma =\Gamma(x,\vm;y,\vxi)$ is a lattice ball that depends on the positions $x,y$ and the oscillator configurations $\vm,\vxi$. When the kernel $K$ satisfies $\sum_{u,\vla} K_\Lambda(x,\vm;u,\vla) < 1$, this bound may be iterated to obtain decay. 

The result for $k=0$ will be established by adapting a similar large disorder argument for the Anderson model (see \cite{Schenker2014a}).  The induction to higher $k$ will be accomplished by an expansion similar to the $k=0$ case.  However, the basic step eq.\ \eqref{eq:basicstep} will be used only if $\|x-y\|> R=R_{\vm|\vxi}(x)$.   For $\|x-y\| \le R$, by restricting the system to a ball of radius $R$ centered at $x$, we will obtain a local volume with excitation number less than $k$ on which the induction hypothesis may be used.

There are three main ingredients in the derivation of the expansion eq. \eqref{eq:basicstep} and the proof of Theorem \ref{thm:main}:
 \begin{enumerate}
 \item Geometric resolvent identities, presented in \S\ref{sec:GRI}, which allow to localize the Green's function to finite regions. 
 \item The \emph{a priori} finiteness of fractional moments, explained in \S\ref{sec:FM}.
 \item Combes-Thomas bounds, presented in \S\ref{sec:CT}, which control decay off the diagonal for the Green's function at energies away from the spectrum of Hamiltonian.	
 \end{enumerate}

\section{Preliminaries} 
\subsection{A word regarding constants}  Throughout we will use the notation $C_{a,b,c,\ldots}$ for an unspecified finite quantity that depends on the indicated parameters $a,b,c,\ldots$ but is otherwise independent of  various other parameters relevant to the discussion.

\subsection{Conditional expectations}
Let $\Omega_\Lambda$ denote the probability space $\Pi_{x\in \Lambda} I_V$ with probability measure $\Pr \ = \ \otimes_{x\in \Lambda} \rho(v_x)\di v_x$. Given $\omega = (v_x)_{x\in \Lambda} \in \Omega_\Lambda$ and $\Gamma \subset \Lambda$, let 
\be \omega_\Gamma \ = \ (v_x)_{x\in \Gamma}\ee
denote the corresponding restriction of the potential to $\Gamma$, which is an element of $\Omega_\Gamma$. Given  $\Gamma \subset \Lambda$, let $\bbE_\Gamma$ denote the operation of \emph{integrating out the variables $(v_x)_{x\in \Gamma}$}. That is  $\bbE_\Gamma$ is  conditional  expectation $\Evc{\cdot}{\Sigma_{\Gamma^c}}$ onto the space of functions measurable with respect to the $\sigma$-field $\Sigma_{\Gamma^c}$ generated by $(v_x)_{x\in \Gamma^c}$.  Put more simply, given $f\in L^1(\Omega)$, 
\be \bbE_\Gamma [ f ] \left ( (v_x)_{x\in \Gamma^c} \right ) \ = \ 
  \int_{\Omega_\Gamma}     f\left ((v_x)_{x \in \Lambda}  \right ) \prod_{x \in \Gamma}  \rho(v_x)\di v_x \ \in \ L^1(\Omega_{\Gamma^c}).
                   \ee

\subsection{Geometry of  $\Lambda$ and $\cE_\Lambda$}  Let $\Gamma\subset \Lambda$ be a subset of $\Lambda$. Given an  oscillator state $\vm\in \cM_\Lambda$,  let $\vm_\Gamma$ denote the restriction of $\vm$ to $\Gamma$.  That is, $\vm_\Gamma \in \cM_{\Gamma}$ and 
\be
\vm_\Gamma(x) \ = \ \vm(x) \quad \text{ for all }x\in \Gamma.
\ee
 We use the following notation for two sets affiliated to $\Gamma$: 
    \be \partial \Gamma \ := \ \setb{x\in \Gamma}{\text{there is } y \in \Lambda \bks \Gamma \textnormal{ so that } \norm{x-y}=1}, \ee
 the boundary of $\Gamma$, and 
 \be \Gamma^\circ \ := \ \Gamma \bks \partial \Gamma,\ee
 the ``interior'' of $\Gamma$. 
 
Given $x\in \Lambda$ and $R \ge0$, let
\be B_{R;\Lambda}(x) \ := \ \setb{u\in \Lambda}{\norm{u-x}\le R},\ee
the ball of radius $R$ centered at $x$.  We note that the numbers of sites in the ball and on its boundary satisfy
\be \abs{B_{R;\Lambda}(x)} \ \le \ \sigma_D (R\vee 1)^D \quad \text{and} \quad \abs{\partial B_{R;\Lambda}(x)} \ \le \ \sigma_D (R\vee 1)^{D-1},\ee
with $\sigma_D$ a $\Lambda$ independent constant.

We make use of two pseudo-metrics on the basis set  $\cE_\Lambda$, defined above in eq.\ \eqref{eq:cELambda}. The first of these, $\Upsilon_\Lambda$, appeared in Thm.\ \ref{thm:main} and was defined in eq.\ \eqref{ups}.
 \bp \label{metric}
    $ \Upsilon_\Lambda$ is a pseudo-metric and  $\Upsilon_\Lambda(x,\vec{m};y,\vec{\xi})=0$ if and only if $x=y$ and $\vec{m}(w)=\vec{\xi}(w)$ for $w\neq x$.      \ep
   \bpf  Clearly $\Upsilon_{\Lambda}$ is non-negative, is symmetric, and vanishes precisely on the claimed set of arguments.  In particular $\Upsilon_{\Lambda}$ vanishes on the diagonal. To verify the triangle inequality, consider points $\ket{x,\vm}$, $\ket{w,\vla}$ and $\ket{y,\vxi}$ in $\cE_\Lambda$.  We consider three cases:
\begin{enumerate}
   \item $\Upsilon_\Lambda(x,\vm;y,\vxi)= \|x-y\|$,
   \item $\Upsilon_\Lambda(x,\vm;y,\vxi) = \|x-u\|$ for some point $u$ such that $\vxi(u)\neq \vm(u)>0$, and
   \item $\Upsilon_\Lambda(x,\vm;y,\vxi) = \|v-y\|$ for some point $v$ such that $\vm(v)\neq \vxi(v)>0$.
\end{enumerate}
In case 1, 
   \be \Upsilon_\Lambda(x,\vm;y,\vxi) \ \le \ \|x-w\| + \|w-y\| \ \le \ \Upsilon_\Lambda(x,\vm;w,\vla) + \Upsilon_\Lambda(w,\vla;y,\vxi).\ee
In case 2, if $\vla(u)\neq \vm(u)$, then
\be \Upsilon_\Lambda(x,\vm;y,\vxi) \ = \ \|x-u\| \ \le \ \Upsilon_\Lambda(x,\vm;w,\vla) \ \le \ \Upsilon_\Lambda(x,\vm;w,\vla) + \Upsilon_\Lambda(w,\vla;y,\vxi).\ee 
On the other hand, if $\vla(u)=\vm(u)$, then $\vla(u)>0$ and $\vla(u)\neq \vxi(u)$ so
\be \Upsilon_\Lambda(x,\vm;y,\vxi) \ = \ \|x-u\| \ \le \ \|x-w\| + \|w-u\| \ \le \ \Upsilon_\Lambda(x,\vm;w,\vla) + \Upsilon_\Lambda(w,\vla;y,\vxi).\ee
Case 3, is similar to case 2 with the roles of $(x,\vm)$ and $(y,\vxi)$ interchanged.
   \epf

Sums of exponentials of the pseudo-metric $\Upsilon_\Lambda$  are bounded by the following
\begin{lemma}\label{lem:expsumbound} For any $\epsilon >0$ there is $C_{\epsilon}<\infty$ such that if $\epsilon <\mu$, then 
	\begin{equation}\label{eq:expsumbound0}
	\sum_{\vla \in \cE_\Lambda^{(k)}} 	\e^{-\mu \Upsilon_\Lambda(x,\vla;y,\vxi)} \ \le \ C_{\epsilon} \e^{-(\mu-\epsilon)\|x-y\|} 
	\end{equation}
	for any $\Lambda \subset \bbZ^d$,  $k=0,1,\ldots$, $x,y\in \Lambda$, and $\vxi\in \cM_\Lambda$.   Thus, given $\mu >0$, $0<\epsilon <\mu$, and $\alpha \in \bbR$ there is $C_{\mu,\epsilon,\alpha} < \infty$ such that
\begin{multline} \label{eq:expsumbound}
 \sum_{\vla \in \cE_\Lambda} \left( N_\Lambda(\vla) \vee 1 \right )^\alpha \e^{-\mu \left  (\Upsilon_\Lambda(x,\vla;y,\vxi) + \abs{\sqrt{N_\Lambda(\vla)}-\sqrt{N_\Lambda}(\vxi)} \right )}  \\ \le \ C_{\mu,\epsilon,\alpha} \left (N_\Lambda(\vxi)\vee 1 \right)^{\alpha + \frac{1}{2}}  \ \e^{-(\mu-\epsilon) \norm{x-y}} \end{multline}
for any $\Lambda \subset \bbZ^d$, $x,y\in \Lambda$, and $\vxi \in \cM_\Lambda$.
\end{lemma}
\begin{proof}  We have 
\begin{multline}
	\sum_{\vla \in \cM_\Lambda^{(k)}}\e^{- \mu \Upsilon_\Lambda(x,\vla;y,\vm)} \ = \ \sum_{R=0}^\infty \sum_{\substack{\vla \in \cM_\Lambda^{(k)} \\ R_{\vla|\vxi}(x) = R}} \e^{- \mu\Upsilon_\Lambda(x,\vla;y,\vm)} \\
	\le \ \sum_{R=0}^{\|x-y\|} \sum_{\substack{\vla \in \cM_\Lambda^{(k)} \\ R_{\vla|\vxi}(x) = R}} \e^{-\mu \|x-y\| } + \sum_{R > \|x-y\|}  \sum_{\substack{\vla \in \cM_\Lambda^{(k)} \\ R_{\vla|\vxi}(x) = R}} \e^{-\mu R}.
\end{multline}
For an oscillator configuration $\vla\in \cM_\Lambda^{(k)}$ with $R_{\vla|\vxi}(x)=R$, any point $u$ at which $\vla(u)\neq 0$ and $\vla(u)\neq \vxi(u)$ must satisfy $\|u-x\|\le R$.  Thus the number of such configurations is bounded by $|B_R(x)|^k \le \ \sigma_D^k (R \vee 1)^{Dk}$.  Therefore,
\begin{multline}
	\sum_{\vla \in \cM_\Lambda^{(k)}}\e^{- \mu \Upsilon_\Lambda(x,\vla;y,\vm)}
	\le \  \sigma_D^k \sum_{R=0}^{\|x-y\|} (R\vee 1)^{Dk} \e^{-\mu \|x-y\| } + \sigma_D^k \sum_{R > \|x-y\|}  (R\vee 1)^{Dk} \e^{-\mu R} \\
	\le C_\epsilon \e^{-(\mu-\epsilon) \|x-y\|}
\end{multline}
with $C_\epsilon <\infty$  independent of $\Lambda$, $k$, $\mu$,  $x$, $y$, and $\vxi$. Thus eq.\ \eqref{eq:expsumbound0} holds. 

Turning now to eq.\ \eqref{eq:expsumbound}, we have
 \begin{multline} \sum_{\vla \in \cE_\Lambda}  \left( N_\Lambda(\vla) \vee 1 \right ) ^\alpha \e^{-\mu \left  (\Upsilon_\Lambda(x,\vla;y,\vxi) + \abs{\sqrt{N_\Lambda(\vla)}-\sqrt{N_\Lambda}(\vxi)} \right )} \\ = \  \sum_{k=0}^\infty (k \vee 1)^\alpha \e^{-\mu\abs{\sqrt{k}-\sqrt{N_\Lambda(\vxi)}}} \sum_{\vla \in \cM_\Lambda^{(k)}}\e^{- \mu\Upsilon_\Lambda(x,\vla;y,\vxi)} \\
 \le \ C_\epsilon \sum_{k=0}^\infty (k \vee 1)^\alpha \e^{-\mu\abs{\sqrt{k}-\sqrt{N_\Omega(\vxi)}}}  \e^{-(\mu-\epsilon)\|x-y\|}
 \end{multline}
 by eq.\ \eqref{eq:expsumbound0}. The result follows since
\be \sum_{k=0}^\infty (k\vee 1)^\alpha \e^{-\mu \abs{\sqrt{k}-\sqrt{n}}} \ \le \ C_{\alpha,\mu} n^{\frac{1}{2}+\alpha},\ee
by Lem.\ \ref{lem:expsum} below. 
\end{proof}

The second, stronger pseudo-metric, $L_\Lambda$, appears in the Combes-Thomas bound stated below.  It is defined as follows.  For $\ket{x,\vm},\ket{y,\vxi}\in \cE_\Lambda$ let
\be  W_{\vm,\vxi} \ = \ \setb{w\in \Lambda}{\vm(w)\neq \vxi(w)} \ee
and let
\begin{multline}\label{eq:LLambda} L_\Lambda(x,\vm;y,\vxi) \ := \ \text{minimal length of a nearest neighbor walk from $x$ to $y$ in $\Lambda$,}\\ \text{visiting all sites in $W_{\vm,\vxi}$.}	
\end{multline}

\begin{proposition}\label{prop:LUps}
$L_\Lambda$ is a pseudo-metric and $L_\Lambda(x,\vm;y,\vxi)=0$ if and only if $x=y$ and $\vm(w)=\vxi(w)$ for $w\neq x$. Furthermore,
\be \label{eq:LUps}L_\Lambda(x,\vm;y,\vxi)\ge \Upsilon_\Lambda(x,\vm;y,\vxi).\ee
\end{proposition}
\begin{proof}
That $L_\Lambda\ge \Upsilon_\Lambda$ is clear from the definitions.  Symmetry and positivity of $L_\Lambda$ are also clear, as is the condition for vanishing of $L_\Lambda$.  To see the triangle inequality, note that if $W_1$ is a minimal length walk from $x$ to $y$ visiting all sites in $W_{\vm,\vxi}$ and $W_2$ is a minimal length walk from $y$ to $z$ visiting all sites in $W_{\vxi,\vla}$ then the concatenation $W$ of $W_1$ and $W_2$ is a walk from $x$ to $z$ which visits every site $u$ such that 
	 $\vm(u)\neq \vxi(u)$ or $\vxi(u) \neq \vla(u).$
	Since one of these two conditions must hold at a site with $\vm(u)\neq \vla(u)$, we see that $W$ visits every site of $W_{\vm,\vla}$ and thus 
	\begin{multline}
		L_{\Lambda}(x,\vm;z,\vla) \ \le \ \text{length of }W \ = \ \text{length of }W_1 + \text{length of }W_2 \\  = \ L_\Lambda(x,\vm;y,\vxi) + L_\Lambda(y,\vxi;z,\vla).\qedhere 
	\end{multline} 
\end{proof}

We define a metric on $\cE_\Lambda$ by
\be \label{eq:dLambda} d_\Lambda(x,\vm;y,\vxi) \ := \ L_\Lambda(x,\vm;y,\vxi) + r_\Lambda(\vm,\vxi).\ee 
Note that this is a metric, since if $L_\Lambda(x,\vm;y,\vxi)=0$ but $\vm\neq \vxi$, then $x=y$ and $\vm(x)\neq \vxi(x)$ so $r_\Lambda(\vm,\vxi)\neq 0$. 
Since $\Upsilon_\Lambda(x,\vla;y,\vxi) + \abs{\sqrt{N_\Lambda(\vla)}-\sqrt{N_\Lambda}(\vxi)} \le d_\Lambda(x,\vm;y,\vxi)$,  Lem.\ \ref{lem:expsumbound}  immediately implies the following
\begin{lemma}\label{lem:expsumbound2}
For any $\mu >0$, $\alpha \in \bbR$ and $0<\epsilon <\mu $ there is $C_{\epsilon,\mu,\alpha} < \infty$ such that
\begin{equation}\label{eq:L+rsumbound}
	\sum_{\vla \in \cE_\Lambda} \left( N_\Lambda(\vla) \vee 1 \right )^\alpha  \e^{-\mu d_\Lambda(x,\vla;y,\vxi) } \
  \le \ C_{\epsilon,\mu,\alpha} \left (N_\Lambda(\vxi)\vee 1 \right)^{\alpha + \frac{1}{2}}  \ \e^{-(\mu-\epsilon) \norm{x-y}}. 
\end{equation}
for any $\Lambda \subset \bbZ^d$, $x,y\in \Lambda$, and $\vxi \in \cM_\Lambda$.
\end{lemma}

\subsection{Geometric Resolvent Identities}\label{sec:GRI}Given  $\cS \subset \cE_\Lambda$ let $P_{\cS}$ denote the projection onto the states in $\cS$, 
\be P_{\cS} \ := \ \sum_{\ket{x,\vm}\in \cS}  \ket{x,\vm}\bra{x,\vm}.\ee
Let
\be H_{\cS} \ = \ P_{\cS} H_\Lambda P_{\cS}, 
\ee
and
\be  \quad G_{\cS}(x,\vec{m};y,\vec{\xi};z) \ := \ \begin{cases} \dirac{x,\vm}{\left(H_{\cS} - z\right)^{-1}}{y,\vxi} & \text{if } \ket{x,\vm} \in \cS \text{ and } \ket{y,\vxi} \in \cS, \\
 0 & \text{otherwise.}	
 \end{cases}
\ee
Note that the Green's function $G_{\cS}$ is defined on all of $\cE_\Lambda^2$, but vanishes off of $\cS^2$. 

In general we shall be interested in two types of subsets $\cS$, those that come from restricting the particle position and those that come from restricting the number of oscillator excitations:

\subsubsection*{Restricting the particle position}  Given $\Gamma \subset \Lambda$ let $P_\Gamma$ denote the projection onto states with the particle in $\Gamma$, 
\be P_\Gamma \ := \ P_{\cE^\Gamma_\Lambda} \ = \  \sum_{\substack{x\in \Gamma \\ \vm \in \cM_\Lambda}}  \ket{x,\vm}\bra{x,\vm} ,\ee
where 
\be \cE^\Gamma_\Lambda \ = \ \setb{\ket{x,\vm}}{x\in \Gamma \quad \text{and} \quad \vm\in \cM_\Lambda}.\ee
Given disjoint sets $\Lambda_1, \Lambda_2 \subset \Lambda$, let
\be    H_{\Lambda_1\oplus \Lambda_2} \ := \ H_{\cE^{\Lambda_1}_\Lambda}+ H_{\cE^{\Lambda_2}_\Lambda} \ = \ 
                P_{\Lambda_1 } H_{\Lambda} P_{\Lambda_1} + P_{\Lambda_2 } H_{\Lambda} P_{\Lambda_2} ,
      \ee
and 
\begin{multline}
	G_{\Lambda_1 \oplus \Lambda_2}(x,\vm;y,\vxi;z) 
     \ := \  \dirac{ x,\vm}{ \left(H_{\Lambda_1\oplus \Lambda_2} - z\right)^{-1}}{y,\vxi}
     \\ = \ G_{\cE^{\Lambda_1}_\Lambda}(x,\vm;y,\vxi;z)  + G_{\cE^{\Lambda_2}_\Lambda}(x,\vm;y,\vxi;z).
\end{multline}  
Furthermore, we set
\be \Delta_{\Lambda_1 \oplus  \Lambda_2}  
                  \ := \ P_{\Lambda_1 }  \Delta_\Lambda P_{\Lambda_2} + P_{\Lambda_2 } \Delta_\Lambda P_{\Lambda_1}  
       \ = \ \Delta_{\Lambda_1} \oplus \Delta_{\Lambda_2},\ee
and define the remainder, 
\be T_{ \Lambda_1;\Lambda_2  } \ := \ \Delta_\Lambda-  \Delta_{ \Lambda_1 \oplus \Lambda_2 } . \ee

The second resolvent identity, applied to $H_{\Lambda_1 \oplus \Lambda_2}$ and $H_{\Lambda}$, yields the following geometric resolvent identities:
\begin{equation}
	\label{gre}
   G_{\Lambda}(x,\vm;y,\vxi;z) \ = \ G_{\Lambda_1 \oplus \Lambda_2}(x,\vm;y,\vxi;z) \ - \gamma \dirac{x,\vm}{\frac{1}{H_{\Lambda_1\oplus \Lambda_2 -z} }T_{\Lambda_1;\Lambda_2} \frac{1}{H_\Lambda -z}}{y,\vxi},
\end{equation}
\begin{equation}
	\label{gre2}
   G_{\Lambda}(x,\vm;y,\vxi;z) 
   =G_{\Lambda_1 \oplus \Lambda_2}(x,\vm;y,\vxi;z) \ - \gamma \dirac{x,\vm}{\frac{1}{H_\Lambda -z}T_{\Lambda_1;\Lambda_2}\frac{1}{H_{\Lambda_1\oplus \Lambda_2 -z} }}{y,\vxi}.  
\end{equation}  

\subsubsection*{Restricting oscillator excitations} Recall that $\cM_{\Lambda}^{(k)}$ denotes the set of all oscillator configurations $\vm$ with total excitation $N_\Lambda(\vm)= k$ \tem see eq.\eqref{eq:cMk}.
Let $\cE_\Lambda^{(k)}$ denote the basis states of $\cH_\Lambda$ with $N_\Lambda(\vm)= k$ ,
\be  \cE_{\Lambda}^{(k)} \ := \ \setb{ \ket{ x, \bm{m}} \in \cE_\Lambda }{  x\in 
	\Lambda  ;  \vm \in \cM_{\Lambda}^{(k)}}.
\ee
The associated projection operators will be denoted as, 
\be  P_\Lambda^{(k;\text{in})} \ := \ P_{\cE_{\Lambda}^{(k)}} \quad 
\text{
and } \quad 
P_\Lambda^{(k;\text{out})} \ := \ 1 - P_\Lambda^{(k)} \ = \ P_{\cE_\Lambda \setminus \cE_\Lambda^{(k)}}.\ee
The $k^{\mathrm{th}}$-band Hilbert space is the range of the $k^{th}$ in-band projection,
\be  \cH_\Lambda^{(k)} \ := \ \ran P_\Lambda^{(k;\text{in})}  \ = \  \overline{\operatorname{span} \, \cE_\Lambda^{(k)}}; \ee
similarly the perpendicular Hilbert space is the range of the complementary operator,
\be  \cH_\Lambda^{(k);\perp} \ := \ \ran P_\Lambda^{(k;\text{out})}  \ = \  \overline{\operatorname{span} \,(\cE_\Lambda \setminus \cE_\Lambda^{(k)})}. \ee

The $k^{\mathrm{th}}$-band restricted Hamiltonian is the Hamiltonian obtained by isolating the $k^{\mathrm{th}}$ band from the remainder of the system:
\be\label{HLambdan}     
H_{\Lambda}^{(k)}  
	\ := \   P_{\Lambda}^{(k;\text{in})} H_{\Lambda}   P_{\Lambda}^{(k;\text{in})}  
       + P_{\Lambda}^{(k;\text{out})} H_{\Lambda}  P_{\Lambda}^{(k;\text{out})} \ = \ H_{\cE_{\Lambda}^{(k)}} + H_{\cE_\Lambda \setminus \cE_\Lambda^{(k)}}.
\ee
The subspaces $\cH_\Lambda^{(k)}$ and $\cH_\Lambda^{(k);\perp}$ are invariant under $H_\Lambda^{(k)}$.  We denote the restrictions of $H_\Lambda^{(k)}$ to these subspaces by
\be  H_\Lambda^{(k;\text{in})} \ := \ \left . H_\Lambda^{(k)}\right |_{\cH_\Lambda^{(k)}} \ = \ H_{\cE_{\Lambda}^{(k)}} \quad \text{and} \quad H_\Lambda^{(k;\text{out})} \ := \ \left . H_\Lambda^{(k)}\right |_{\cH_\Lambda^{(k);\perp}} \ = \ H_{\cE_\Lambda \setminus \cE_\Lambda^{(k)}}.\ee
Note that
\be \sigma(H_\Lambda^{(k;\text{in})}) \subset \bI_k \quad \text{and} \quad \sigma(H_\Lambda^{(k;\text{out})}) \subset \bigcup_{j\neq k} \bI_j,\ee
where $\bI_k$ denotes the $k^{\text{th}}$ band as defined in eq.\ \eqref{eq:bands} (see Theorem \ref{thm:specCT} below).  

The Green's functions corresponding to the $k^{th}$-band restricted Hamiltonian is 
\begin{multline}
	G_{\Lambda}^{(k)}(x,\vm;y,\vxi;z) 
	\ := \  \dirac{ x,\vm }{ \left(H_{\Lambda}^{(k)} - z\right)^{-1}}{y,\vxi} \\ = \ G_{\Lambda}^{(k;\text{in})}(x,\vm;y,\vxi;z) + G_{\Lambda}^{(k;\text{out})}(x,\vm;y,\vxi;z) ,
\end{multline} where we have introduced the Greens functions of $H_\Lambda$ restricted to the $k^{th}$ band projection and its complement:
\be G_{\Lambda}^{(k;\text{in})}(x,\vm;y,\vxi;z)
	\ :=  \ \dirac{ x,\vm }{ \left(H_{\Lambda}^{(k;\text{in})} - z\right)^{-1}}{y,\vxi} \ = \ G_{\cE^{(k)}_\Lambda}(x,\vm;y,\vxi;z)
\ee
and 
\be G_{\Lambda}^{(k;\text{out})}(x,\vm;y,\vxi;z)
	\ :=  \ \dirac{ x,\vm }{ \left(H_{\Lambda}^{(k;\text{out})} - z\right)^{-1}}{y,\vxi} \ = \ G_{\cE_\Lambda \setminus \cE_\Lambda^{(k)}}(x,\vm;y,\vxi;z) .
\ee

Note that the $k^{th}$ band restricted Hamiltonian is of the form \be  H_{\Lambda}^{(k)} \ = \ \gamma  \Delta_{\Lambda}^{(k)}  + \omega H_{\mathrm{ph}} + V_\Lambda ,\ee
with the reduced hopping operator,
\be  \Delta_{\Lambda}^{(k)} 
	= P_{\Lambda}^{(k)}  \Delta_\Lambda  P_{\Lambda}^{(k)}                
	+   P_{\Lambda}^{(k);\perp} \Delta_\Lambda P_{\Lambda}^{(k);\perp}.  
\ee
The remainder of the Laplacian
\be  T_\Lambda^{(k)} \ := \ \Delta_\Lambda- \Delta_\Lambda^{(k)} \ = \ H_\Lambda - H_{\Lambda}^{(k)}\ee
is the hopping operator connecting the $k^{\mathrm{th}}$ band with the rest of the system. Analogous to eqs.\ (\ref{gre}, \ref{gre2}) we have,
\begin{equation}
\label{firstres}
G_{\Lambda}(x,\vm;y,\vxi;z) 
	\ = \ G_{\Lambda}^{(k)}(x,\vm;y,\vxi;z)\\ 
	 - \gamma
	 \dirac{x,\vm}{\frac{1}{H_\Lambda^{(k)}-z}T_{\Lambda}^{(k)} \frac{1}{H_\Lambda -z}}{y,\vxi} 
\end{equation}
and 
\begin{equation}
	\label{firstres2}
G_{\Lambda}(x,\vm;y,\vxi;z) 
	\ = \ G_{\Lambda}^{(k)}(x,\vm;y,\vxi;z)
	 - \gamma
	 \dirac{x,\vm}{ \frac{1}{H_\Lambda -z}T_{\Lambda}^{(k)}\frac{1}{H_\Lambda^{(k)}-z}}{y,\vxi}	
\end{equation}

A typical application of eq.\ \eqref{firstres} will involve states $\ket{x,\vm}\not \in \cE_\Lambda^{(k)}$ and $\ket{y,\vxi}  \in \cE_{\Lambda}^{(k)}$.  In this case, the first term on the right hand side vanishes (since $\cH_{\Lambda}^{(k)}$ is invariant under $H_\Lambda^{(k)}$) and eq.\ \eqref{firstres} reduces to
\begin{equation}
	\label{firstresout}
G_{\Lambda}(x,\vm;y,\vxi;z)  \  = \ - \gamma
	 \dirac{x,\vm}{\frac{1}{H_\Lambda^{(k;\text{out})}-z}T_{\Lambda}^{(k)} \frac{1}{H_\Lambda -z}}{y,\vxi}. 
\end{equation}

\subsection{Fractional moment bound} \label{sec:FM}
The proof of Theorem \ref{thm:main} relies on two  \emph{a priori} estimates on the Green's function.  The first of these is the \emph{fractional moment bound}, generalized to this context from the adaptation of the Aizenman-Molchanov approach to continuum Schr\"odinger operators \cite{Aizenman2006}. Although the Green's function $G_{\cS}(x,\vm;y,\vxi;z)$ is singular as $\Im z \rightarrow 0$, the disorder in the random potential mollifies averages of the Green's function raised to a fractional power:
\def\fmc{\kappa}
\begin{lemma}\label{lem:apriori}
There is $\kappa <\infty$ such that for any $s< 1$,
\be \label{eq:fmb} \bbE_{\{x,y\}}(\abs{G_{\cS}(x,\vm;y,\vxi;z)}^s) \ \le \ \frac{\fmc}{1-s}\ee
for any $\cS \subset \cE_\Lambda$.
\end{lemma}

We prove this bound in Appendix \ref{sec:fmt} below. One key, if simple, corollary of this result is the following ``all-for-one lemma'' which shows that a bound on one fractional moment bounds all other fractional moments: 
\begin{lemma}\label{lem:allforone} If  for some $0<s_0<1$ we have the bound
\be \Ev{\abs{G_{\mc{S}}(x,\vm;y,\vxi,z)}^{s_0}} \ \le \ A \e^{- \Phi(x,\vm;y,\vxi)}\ee
with $A<\infty$ and $\Phi$ a real valued function on $\mc{S}\times \mc{S}$, 
then for all $0<s <1$ we have
\be \Ev{\abs{G_{\mc{S}}(x,\vm;y,\vxi,z)}^{s}} \ \le \ A_s \e^{-\alpha_s\Phi(x,\vm;y,\vxi)}\ee
with $A_s <\infty $ and $\alpha_s>0$ continuous functions of $s$ on $(0,1)$.
\end{lemma}
\begin{proof}
This follows  from H\"older's inequality.  There are two cases:
\begin{enumerate}
	\item If $s\le s_0$, then $\Ev{\abs{G_{\mc{S}}(x,\vm;y,\vxi,z)}^{s}} \ \le \ \Ev{\abs{G_{\mc{S}}(x,\vm;y,\vxi,z)}^{s_0}}^{\frac{s}{s_0}} \ \le \ A^{\frac{s}{s_0}} \e^{-\frac{s}{s_0} \Phi(x,\vec{m};y,\vec{\xi})}.$
	\item If $s_0 <s <1$, then
	\begin{multline} \Ev{\abs{G_{\mc{S}}(x,\vm;y,\vxi,z)}^{s}} \ \le \ \Ev{\abs{G_{\mc{S}}(x,\vm;y,\vxi,z)}^{s_0}}^{\frac{r-s}{r-s_0}} \Ev{\abs{G_{\mc{S}}(x,\vm;y,\vxi,z)}^{r}}^{\frac{s-s_0}{r-s_0}}
	\\ \le \ 	\left [ \frac{\kappa}{1-r} \right ]^{\frac{s-s_0}{r-s_0}} A^{\frac{r-s}{r-s_0}} \e^{-\frac{r-s}{r-s_0} \Phi(x,\vm;y,\vxi)},
	\end{multline}
	for any $r\in (s,1)$.  To obtain an explicit bound we may take $r=s + \frac{1-s}{2}$. \qedhere
\end{enumerate}
\end{proof}

\subsection{Combes-Thomas bound}\label{sec:CT}
The second \emph{a prior} estimate is a Combes-Thomas type bound on the Green's function depending on the distance from $z$ to the spectrum of the operator.,The general estimate, Thm.\ \ref{thm:CT}, is stated and proved in App.\ \ref{CT section} below.  Here we derive as a corollary of that result the bounds used in the proof of Thm.\ \ref{thm:main}.  

Recall the definition \eqref{eq:dLambda} of the metric $d_\Lambda$ on $\cE_\Lambda$.
For $\cS_1,\cS_2 \subset \cE_\Lambda$, we take \be  d_\Lambda(\cS_1,\cS_2) \ = \inf_{\substack{\ket{x,\vm}\in \cS_1\\ \ket{y,\vxi}\in \cS_2}} d_\Lambda(x,\vm;y,\vxi).\ee
Then we have the following
\begin{theorem}
	\label{thm:specCT} Suppose $\delta=\omega- V_+- 4D\gamma >0$ and  that $\cS\subset \cE_\Lambda$ is such that $\cS\cap \cE_\Lambda^{(k)} = \emptyset$.  Then $\sigma(H_{\cS}) \cap \bI_k = \emptyset$ and  there is $\nu_{k;\gamma} >0$ such that for $\cS_1,\cS_2 \subset \cS$ and $z = E+\im \epsilon$ with $E\in \bI_k$ we have \be  \norm{P_{\cS_1} \frac{1}{H_{\cS} -z} P_{\cS_2}} \ \le \  \frac{2}{\delta} e^{-\nu_{k;\gamma} d_\Lambda(\cS_1,\cS_2)}. \ee
Furthermore, the constant of exponential decay $\nu_{k;\gamma}$ may be chosen so that for any $\mu >0$ there is a constant $C_\mu  <\infty$ such that
\be \label{eq:nubound} \nu_{k;\gamma}  \ \ge  \ \min \left ( \mu , \ C_\mu \frac{\delta }{\gamma} \frac{1}{k+1} \right ).\ee
\end{theorem}
\begin{remark}\label{rem:specCT}
	A key consequence of eq.\ \eqref{eq:nubound} is that, given $\mu>0$ there is a constant $\wt{C}_\mu$ such that
	\be  \norm{P_{\cS_1} \frac{1}{H_{\cS} -z} P_{\cS_2}} \ \le \  \frac{2}{\delta} e^{-\mu d_\Lambda(\cS_1,\cS_2)} \ee 
	provided \be  \gamma \le \wt{C}_\mu \frac{\omega-V_+}{k+1}.\ee
	Indeed, we may take $\wt{C}_\mu = \frac{C_\mu}{4D C_\mu +\mu}.$
\end{remark}
\begin{proof}
This follows from the remark following Thm.\ \ref{thm:CT}, since $E < (k+1)\omega$ for $E\in \bI_k$.
 \end{proof}

\begin{corollary}\label{cor:specCT}Suppose  $\delta=\omega- V_+- 4D\gamma >0$ and that $\cS\subset \cE_\Lambda$  is such that $\cS\cap \cE_\Lambda^{(k)} = \emptyset$ and let $\nu=\nu_{\gamma,k}$ be as in Thm.\ \ref{thm:specCT}. Then there is $A_\nu <\infty$ such that for $z=E + \im \epsilon$ with $E\in \bI_k$ and $\cS_1,\cS_2 \subset \cS$ we have
\be \label{eq:specCT1} \abs{\dirac{x,\vm}{\Delta_\Lambda P_{\cS_1} \frac{1}{H_{\cS}-z} } {y,\vxi}} \ \le \ \frac{A_\nu}{\delta} Q(x,\vm) \e^{-\nu d_\Lambda(x,\vm;y,\vxi)} \ee
and
\be \label{eq:specCT2} \abs{\dirac{x,\vm}{\Delta_\Lambda P_{\cS_1} \frac{1}{H_{\cS}-z} P_{\cS_2}\Delta_\Lambda}{y,\vxi}} \ \le \ \frac{A_\nu}{\delta} Q(x,\vm)Q(y,\vxi) \e^{-\nu d_\Lambda(x,\vm;y,\vxi)} \ , \ee
where 
\be  Q(x,\vm) \ := \ 1 + \left ( \vm(x) \vee 1 \right )^{\frac{1}{4}} \sum_{u\sim x} \left (\vm(u) \vee 1\right )^{\frac{1}{4}}.\ee
\end{corollary}

\begin{proof}
The idea is to combine Prop.\ \ref{disp0} with Thm.\ \ref{thm:specCT}. We will prove only  eq.\ \eqref{eq:specCT1}; the proof of eq.\ \eqref{eq:specCT2} involves an additional  application of Prop.\ \ref{disp0} to bound the sum over matrix elements of $\Delta_\Lambda$ on the right.

By eq.\ \eqref{Kdef},
\begin{multline}
	\abs{\dirac{x,\vm}{\Delta_\Lambda P_{\cS_1} \frac{1}{H_{\cS}-z}}{y,\vxi}} \ \le \ \sum_{u,\vla \in \cE_\Lambda } \abs{K_{x,\vm}^{u,\vla}} \abs{G_{\cS}(u,\vla;y,\vxi)} \\
	\le \ 2D  \abs{G(x,\vm;y,\vxi)} 
	+ \sum_{u \sim x} \sum_{\vla} \abs{K_{x,\vm}^{u,\vla}}\abs{G_{\cS}(u,\vla;y,\vxi)}  .
\end{multline}
Inserting the Combes-Thomas bound for the Green's functions on the right hand side yields,
\begin{multline}
\abs{\dirac{x,\vm}{\Delta_\Lambda P_{\cS_1} \frac{1}{H_{\cS}-z}}{y,\vxi}} \\
\le \  \frac{4D}{\delta} \e^{-\nu d_\Lambda(x,\vm;y,\vxi)} \left ( 1 + \frac{1}{2D} \sum_{u \sim x} \sum_{\vla} \abs{K_{x,\vm}^{u,\vla}} \e^{-\nu d_\Lambda(x,\vm;u,\vla)} \right ).
\end{multline}
The hopping term $K_{x,\vm}^{u,\vla}$ enforces equality between $\vm$ and $\vla$ everywhere expect at $x$ and $u$.  Thus by Prop.\ \ref{disp0}
\begin{multline}
\sum_{\vla} \abs{K_{x,\vm}^{u,\vla}}\e^{\nu d_\Lambda(x,\vm;u,\vla)} \\
\le \ \e^{\nu} \sum_{\lambda_1} \e^{\nu\abs{\sqrt{\vm(x)}-\sqrt{\lambda_1}}}\abs{\dirac{\vm(x)}{D(\beta)}{\lambda_1}}\sum_{\lambda_2} \e^{\nu\abs{\sqrt{\vm(u)}-\sqrt{\lambda_2}}} \abs{\dirac{\vm(u)}{D(\beta)}{\lambda_2}} \\
	 \le \ \e^\nu C_{\nu,\beta}^2 \left ( \vm(x)\vee 1 \right )^{\frac{1}{4}}  \left ( \vm(u)\vee 1 \right )^{\frac{1}{4}} ,
\end{multline}
with $C_{\nu,\beta} <\infty.$ Eq.\ \eqref{eq:specCT1} follows.
\end{proof}

\section{The lowest band}
In this section we prove the $k=0$ case of Thm.\ \ref{thm:main}:
\begin{lt0}
  For any $\mu > 0$ and $0 <s < 1 $ there is $\gamma_{0,\mu,s}$  so that if $\gamma < \gamma_{0,\mu,s}$, then there is $A_{0,\mu,s} <\infty$ such that 
   \be  \Ev{\abs{G_\Lambda(x,\vm;y,\vxi,z)}^s} \ \le \ A_{0,\mu,s} \e^{-\mu \Upsilon_\Lambda (z,\vm;y,\vxi) -\mu r_\Lambda(\vm,\vxi)} \ee 
   for any  $\Lambda \subset \Z^d$,  $|x,\vm\rangle, |y,\vxi\rangle \in  \cE_{\Lambda}$,  and   $z\in \bS_0$.
\end{lt0}
\noindent This will be accomplished in two steps. First, we consider ``in-band correlations'' between states without oscillator excitations:
\begin{lemma}[In-Band] \label{lem:k=0}For any $\mu>0$ and $s<1$ there is $\gamma_{0,\mu,s} >0$ such that if $\gamma <\gamma_{0,\mu,s}$, then 
     \be\label{eq:k=0propbound}
    \Ev{\abs{G_\Lambda(x,\vec{0};y,\vec{0};z)}^s} \ \le \ \frac{\kappa}{1-s} \e^{-\mu \|x-y\|}
   \ee 
for any $\Lambda \subset \Z^d$,   $x,y\in \Lambda$ and  $z\in\bS_0$.
\end{lemma}
\noindent Second, we use the Combes-Thomas bound Cor.\ \ref{cor:specCT} to extend the in-band estimate to ``out-of-band correlations'' between states with arbitrary oscillator excitations and thereby prove Thm.\ L($0$).
 
\subsection{In band correlations}  In this section we prove Lem.\ \ref{lem:k=0},  following the large disorder argument for random Schr\"odinger operators described in \cite{Schenker2014a} (see also \cite[Chapter 6]{aizenman2015random}). Throughout the proof, the energy $z$ will be fixed, and thus we drop this argument from the Green's functions to lighten the notation.

Note that, eq.\ \eqref{eq:k=0propbound} with $x=y$  follows from the \emph{a priori} bound of Lem.\ \ref{lem:apriori}. 
For $x\neq y$, we  apply the geometric resolvent identity \eqref{gre2} with $\Lambda_1 =\{x\}$ to conclude 
\begin{multline}\abs{G_\Lambda(x,\vec{0};y,\vec{0})} \ \le \ \gamma 
\abs{ \dirac{x,\vec{0}} {\frac{1}{H_\Lambda -z} T_{\{x\};\Lambda\setminus \{x\}} \frac{1}{H_{\Lambda \setminus \{x\}} -z }}{y,\vec{0}}} \\
\le \ \gamma \sum_{\vm \in \cM_\Lambda} 
\abs{ \dirac{x,\vec{0}} {\frac{1}{H_\Lambda -z}}{x,\vm}} \abs{\dirac{x,\vm}{T_{\{x\};\Lambda\setminus \{x\}}  \frac{1}{H_{\Lambda \setminus \{x\}} -z }}{y,\vec{0}}}. 
\end{multline}
Raising both sides to the power $s<1$, using the inequality $(a+b)^s \le a^s + b^s$, averaging with respect to the potential $v_x$, and applying the \emph{a priori} bound, Lem.\ \ref{lem:apriori}, yields
\begin{equation}
	\bbE_x\left(\abs{G_\Lambda(x,\vec{0};y,\vec{0})}^s \right ) \ \le \ 
	\gamma^s \frac{\fmc}{1-s} \sum_{\vm \in \cM_\Lambda} 
 \abs{\dirac{x,\vm}{T_{\{x\};\Lambda\setminus \{x\}}  \frac{1}{H_{\Lambda \setminus \{x\}} -z }}{y,\vec{0}}}^s. \label{eq:k=0bound1}
\end{equation}

The terms in the sum on the right hand side can be bounded as follows
\begin{multline}\label{eq:k=0bound2}
	\abs{\dirac{x,\vm}{T_{\{x\};\Lambda\setminus \{x\}}  \frac{1}{H_{\Lambda \setminus \{x\}} -z }}{y,\vec{0}}}^s \ \le \ \abs{\dirac{x,\vm}{T_{\{x\};\Lambda\setminus \{x\}} P^{(0;\text{in})} \frac{1}{H_{\Lambda \setminus \{x\}} -z }}{y,\vec{0}}}^s \\+\abs{\dirac{x,\vm}{T_{\{x\};\Lambda\setminus \{x\}} P^{(0;\text{out})} \frac{1}{H_{\Lambda \setminus \{x\}} -z }}{y,\vec{0}}}^s, 
\end{multline}
where $P^{(0;\text{in})}$ is the ``in-band'' projection onto states with zero oscillator excitation and $P^{(0;\text{out})}$ is the complementary projection.  The ``in-band'' term is bounded by
\be \abs{\dirac{x,\vm}{T_{\{x\};\Lambda\setminus \{x\}} P^{(0;\text{in})} \frac{1}{H_{\Lambda \setminus \{x\}} -z }}{y,\vec{0}}}^s \ \le \ \sum_{x'\sim x} 
\abs{K_{x,\vm}^{x',\vec{0}}}^s \abs{G_{\Lambda\setminus \{x\}}(x',\vec{0};y,\vec{0})}^s,\ee
since $T_{\{x\};\Lambda\setminus\{x\}}$ involves hopping only to the neighbors of $x$. The hopping term $K_{x,\vm}^{x',\vec{0}}$ forces $\vm$ to vanish except at $x$ and $x'$.  Thus
\begin{multline}\label{eq:inbound} \sum_{\vm}\abs{\dirac{x,\vm}{T_{\{x\};\Lambda\setminus \{x\}} P^{(0;\text{in})} \frac{1}{H_{\Lambda \setminus \{x\}} -z }}{y,\vec{0}}}^s \\ \le \ \sum_{x'\sim x}  \sum_{m_1=0}^\infty \sum_{m_2=0}^\infty \abs{\dirac{m_1}{D(\beta)}{0}}^s \abs{\dirac{m_2}{D(\beta)}{0}}^s\abs{G_{\Lambda\setminus \{x\}}(x',\vec{0};y,\vec{0})}^s \\ 
\le \ C_{s} \sum_{x'\sim x}	\abs{G_{\Lambda\setminus \{x\}}(x',\vec{0};y,\vec{0})}^s,
\end{multline}
by Prop.\ \ref{disp0}.

To bound the out-of-band term in eq.\ \eqref{eq:k=0bound2}, we use eq.\ \eqref{firstresout} to return to the $k=0$ band:
\begin{multline}
	\abs{ \dirac{x,\vm} { T_{\{x\};\Lambda\setminus \{x\}} P^{(0;\text{out})} \frac{1}{H_{\Lambda \setminus \{x\}} -z }}{y,\vec{0}}}^s \\ \le \ \gamma^s\abs{ \dirac{x,\vm} { T_{\{x\};\Lambda\setminus \{x\}} P^{(0;\text{out})} \frac{1}{H_{\Lambda \setminus \{x\}}^{(0;\text{out})} -z } T_\Lambda^{(0)}P^{(0;\text{in})} \frac{1}{H_{\Lambda \setminus \{x\}}} }{y,\vec{0}}}^s \\
	\le \ \gamma^s \sum_{x'\in \Lambda \setminus \{x\}} 
	\abs{ \dirac{x,\vm} { T_{\{x\};\Lambda\setminus \{x\}} P^{(0;\text{out})} \frac{1}{H_{\Lambda \setminus \{x\}}^{(0;\text{out})} -z } T_\Lambda^{(0)}}{x',\vec{0}}}^s \abs{G_{\Lambda\setminus \{x\}}(x',\vec{0};y,\vec{0})}^s.
\end{multline}
By Cor.\ \ref{cor:specCT} there is $\wt{\gamma}_0$ so that for $\gamma < \wt{\gamma}_0$, we have
\begin{multline}\label{eq:outbound}
	\abs{ \dirac{x,\vm} { T_{\{x\};\Lambda\setminus \{x\}} P^{(0;\text{out})} \frac{1}{H_{\Lambda \setminus \{x\}} -z }}{y,\vec{0}}}^s \\ \le \ C_{s,\mu}  \gamma^s Q(x,\vm)^s\sum_{x'\in \Lambda \setminus \{x\}} \e^{-3\mu d_\Lambda(x,\vm;x',\vec{0})} \abs{G_{\Lambda\setminus \{x\}}(x',\vec{0};y,\vec{0})}^s,\end{multline}
where
$Q(x,\vec{m}) \ = \ 1 + \sum_{u\sim x} \left ((\vec{m}(x)\vee 1)(\vec{m}(u)\vee 1)\right )^{\frac{1}{4}}.$ 

Since $Q(x,\vec{m}) \ \le \ C (N_\Lambda(\vm)\vee 1)^{\frac{1}{2}}$, Lem.\ \ref{lem:expsumbound2} implies
\begin{equation}
	\sum_{\vec{m}} Q(x,\vec{m})^s\e^{-3\mu d_\Lambda(x,\vec{m};\vec{x'},\vec{0})}
	\ \le \ C_{\mu,s} \e^{-2\mu \|x-x'\|}.
\end{equation}
Thus 
\begin{multline}\sum_{\vec{m}} \abs{\dirac{x,\vm}{T_{\{x\};\Lambda\setminus \{x\}} P^{(0;\text{out})} \frac{1}{H_{\Lambda \setminus \{x\}} -z }}{y,\vec{0}}}^s \\ \le \ C_{s,\mu}\gamma^s\sum_{x'\in \Lambda\setminus \{x\}} \e^{-2\mu\|x-x'\|} \abs{G_{\Lambda\setminus \{x\}}(x',\vec{0};y,\vec{0})}^s ,
\end{multline}
provided $\gamma <\wt{\gamma_0}$. Together with eq.\ \eqref{eq:inbound} this implies, using eqs.\ \eqref{eq:k=0bound2} and \eqref{eq:k=0bound1}, that
\begin{equation}\label{eq:k=0iterbound}\bbE_x\left(\abs{G_\Lambda(x,\vec{0};y,\vec{0})}^s \right ) \ \le \ 
	\sum_{x'\in \Lambda\setminus \{x\}} K_{\gamma,s} (x,x') \abs{G_{\Lambda\setminus\{x\}}(x',\vec{0};y,\vec{0})}^s\end{equation}
where
\be K_{\gamma,s}(x,x') \ = \ C_s \gamma^s  I[x\sim x'] + C_{s,\mu}\gamma^{2s} \e^{-2\mu\|x-x'\|}.\ee

Eq.\ \eqref{eq:k=0iterbound} may be iterated to prove Lem.\ \ref{lem:k=0}.  An efficient way to proceed is to fix a finite set $\Lambda \subset \bbZ^d$ 
\be\label{eq:iterator} F_\Lambda \ = \ \max_{W\subset \Lambda} \max_{x,y\in W} \e^{\mu \|x-y\|} \Ev{\abs{G_\Lambda(x,\vec{0};y,\vec{0})}^s }.\ee
For $x\neq y$, 
\begin{multline}\label{eq:notmax}
	\e^{\mu \|x-y\|} \Ev{\abs{G_W(x,\vec{0};y,\vec{0})}^s } \ \le \  \sum_{x'\neq x} K_{\gamma,s}(x,x')\e^{\mu\|x-x'\|} \e^{\mu\|x'-y\|} \Ev{\abs{G_{W \setminus\{x\}} (x',\vec{0};y,\vec{0})}^s }\\
	\le \ F_\Lambda \sum_{x'\neq x} K_{\gamma,s}(x,x')\e^{\mu\|x-x'\|} .
\end{multline}
If 
$\sup_x \sum_{x'}K_{\gamma,s}(x,x') \e^{\mu\|x-x'\|} < 1$, since $F_\Lambda <\infty$ because $\Lambda$ is finite, eq.\ \eqref{eq:notmax} implies that the maximum in eq.\ \eqref{eq:iterator} cannot occur for $x\neq y$ and thus that 
\be  \e^{\mu \|x-y\|} \Ev{\abs{G_\Lambda(x,\vec{0};y,\vec{0})}^s} \ \le \ \max_{W\subset \Lambda} \max_{x\in W} \Ev{\abs{G_\Lambda(x,\vec{0};x,\vec{0})}^s } \ \le \ \frac{\kappa}{1-s},\ee
by the \emph{a priori} bound, Lem.\ \ref{lem:apriori}. Since 
\be \sup_x \sum_{x'}K(x,x')\e^{\mu \|x-x'\|}  \ \le  \ C_s \gamma^s + C_{s,\mu} \gamma^{2s} ,\ee 
we see that
\be \Ev{\abs{G_\Lambda(x,\vec{0};y,\vec{0})}^s} \ \le \ \frac{\kappa}{1-s} \e^{-\mu \|x-x'\|} \ee 
provided  $\gamma <\wt{\gamma}_0$ and
\be  C_s \gamma^s + C_{s,\mu} \gamma^{2s}  < 1.\ee

As the bounds obtained are independent of the finite set $\Lambda$, they extend \emph{a posteriori} to infinite $\Lambda$ by strong resolvent convergence.
This completes the proof of Lem.\ \ref{lem:k=0}.

\subsection{Out of band correlations: the proof of Thm.\ L($0$)} In this section we prove Thm.\ L($0$). The result follows directly from Lem.\ \ref{lem:k=0} for  $\ket{x,\vm}$, $\ket{y,\vxi} \in  \cE_{\Lambda}^{(0)}$.
  If one or both of the terms $\ket{x,\vm}$, $\ket{y,\vxi}$ is not in the band $\cE_{\Lambda}^{(0)}$, we will apply a geometric resolvent identity and the Combes-Thomas bound, Cor.\ \ref{cor:specCT}, to obtain the required estimate.
    
  Let us first assume $\ket{x,\vm} = \ket{x,\vo}\in  \cE_{\Lambda}^{(0)}$
 and $\ket{y,\vxi} \not \in \cE_{\Lambda}^{(0)}$.
 One application of (\ref{firstres}) yields,
 \be \label{FM0e1}
 |G_\Lambda( x,\vo; y,\vxi ) |^s 
    \leq \gamma^s \sum_{w} 
    | G_\Lambda(  x,\vo ; w,\vo )|^s
    \abs{\dirac{w,\vo}{T_\Lambda^{(0)} \frac{1}{H_\Lambda^{(0)} - z}}{y,\vxi}}^s
  \ee
  By the Combes-Thomas bound, Cor.\ \ref{cor:specCT}, there is $\alpha_\mu>0$ so that if $\gamma <\alpha_\mu$, then 
  \be |G_\Lambda( x,\vo; y,\vxi ) |^s 
    \leq  C_{\mu,s}  \sum_{w} 
    | G_\Lambda(  x,\vo ; w,\vo )|^s \e^{- 2 \mu  d_\Lambda(w,\vo;y,\vxi)}. \ee
    By the in-band lemma, Lem.\ \ref{lem:k=0},
    \be \Ev{|G_\Lambda( x,\vo; y,\vxi ) |^s} 
    \leq  C_{\mu,s} \frac{\kappa^s}{1-s}  \sum_{w} 
    \e^{-\mu \|x-w\|}   \e^{- 2\mu d_\Lambda(w,\vo;y,\vxi) }.\ee
    Recall that $d_\Lambda(w,\vo;y,\vxi)  =  L_\Lambda(w,\vo;y,\vxi) + r_\Lambda(\vo,\vxi) $,
    where $L_\Lambda\ge \Upsilon_\Lambda$ \tem see Prop.\ \ref{prop:LUps}. Thus
     \begin{multline}
     	\Ev{|G_\Lambda( x,\vo; y,\vxi ) |^s} 
    \leq \  C_{\mu,s}   \sum_{w} 
    \e^{-\mu \|x-w\|}   \e^{- 2\mu \Upsilon_\Lambda(w,\vo;y,\vxi)} \e^{-2\mu  r_\Lambda(\vo,\vxi) } \\ 
    \leq 
    C_{\mu,s} \e^{-\mu \Upsilon_\Lambda(x,\vo;y,\vxi)- \mu r_\Lambda(\vo,\vxi) }  \sum_{w} \e^{- \mu \|w-y\| } \\ \le \ C_{\mu,s} \e^{- \mu  \Upsilon_\Lambda(x,\vo;y,\vxi) -\mu  r_\Lambda(\vo,\vxi) }, 
     \end{multline} 
    where we have applied the triangle inequality for $\Upsilon_\Lambda$ and the bound $\Upsilon_\Lambda(w,\vo;y,\vxi)\ge \|w-y\|$. 
    
    If $\ket{y,\vxi}=\ket{y,\vo}$ and $\ket{x,\vm}\not \in \cE_\Lambda^{(0)}$, the argument is similar, but involves the geometric resolvent identity eq.\ \eqref{firstres2}.
    
    Finally consider both $\vm, \vxi \notin \cM_\Lambda^{(0)}$.
    Now an application of eq.\ (\ref{firstres}) and eq.\ \eqref{firstres2} yields,
    \begin{multline}
    	\abs{G_\Lambda ( x,\vm; y,\vxi )} ^s \ \le \  \abs{G^{(0)}_\Lambda( x,\vm; y,\vxi ) }^s 
    	 \\ + \ \gamma^{2s} 
       \sum_{w_1,w_2}\abs{\dirac{x,\vm}{ \frac{1}{H_\Lambda^{(0)} - z}T_\Lambda^{(0)} }{w_1,\vo}}^s	    	\abs{G_\Lambda ( w_1,\vo; w_2,\vo )} ^s 
       \abs{\dirac{w_2,\vo}{T_\Lambda^{(0)} \frac{1}{H_\Lambda^{(0)} - z}}{y,\vxi}}^s .
    \end{multline}
    By Thm.\ \ref{thm:specCT} and Cor.\ \ref{cor:specCT}, if $\gamma <\alpha_\mu$ then 
    \begin{multline}
    \abs{G_\Lambda(x,\vm;y,\vxi)}^s \ \le \ C_{\mu,s} \e^{-2\mu  d_\Lambda(x,\vm;y,\vxi)} \\ + C_{\mu,s}\sum_{w_1,w_2} \e^{-2 \mu d_\Lambda(x,\vm;w_1,\vo)}\abs{G_\Lambda ( w_1,\vo; w_2,\vo )} ^s \e^{-2\mu d_\Lambda(w_2,\vo;y,\vxi)}
    \end{multline}
    Arguing as above, we find that 
    \begin{multline}
    \Ev{\abs{G_\Lambda(x,\vm;y,\vxi)}^s} \\ \le \ C_{\nu,s} \e^{-\mu \Upsilon_\Lambda(x,\vm;y,\vxi)-\mu r_\Lambda(\vm,\vo)-\mu r_\Lambda(\vo,\vxi)}  \left [ 1 + \sum_{w_1,w_2} \e^{-\mu \|x-w_1\|} \e^{-\mu\|y-w_2\|}  \right ].
    \end{multline}
    Using the triangle inequality for $r_\Lambda$, this implies 
    \be \Ev{\abs{G_\Lambda(x,\vm;y,\vxi)}^s} \\ \le \ C_{\mu,s} \e^{-\mu \Upsilon_\Lambda(x,\vm;y,\vxi)-\mu  r_\Lambda(\vm,\vxi)}. \ee
    This completes the proof of Thm.\ L($0$).

\section{Higher bands}
In this section we prove Thm.\ \ref{thm:main}.  That is we prove for each $k=0,1,2,\ldots$,
\begin{lt} 
  For any $\mu > 0$ and $0 <s < 1 $ there is $\gamma_{k,\mu,s}$  so that if $\gamma < \gamma_{k,\mu,s}$, then there is $A_{k,\mu,s} <\infty$ such that 
   \be  \Ev{\abs{G_\Lambda(x,\vm;y,\vxi,z)}^s} \ \le \ A_{k,\mu,s} \e^{-\mu \Upsilon_\Lambda (z,\vm;y,\vxi) -\mu \cR^{(k)}_\Lambda(\vm,\vxi)} \ee 
   for any  $\Lambda \subset \Z^d$,  $|x,\vm\rangle, |y,\vxi\rangle \in  \cE_{\Lambda}$,  and   $z\in \bS_k$.
\end{lt}

We have already established Theorem L($0$) in the previous section.  The remaining cases will be established by induction on $k$.  Specifically, for each $k$, let Condition C($k$) refer to 
\begin{ct}
	Theorem L($j$) holds for $j=0,\ldots,k-1$.
\end{ct}  
\noindent We prove Thm.\ \ref{thm:main} by deriving the implication  ``Condition C($k$) $\Rightarrow$ Theorem L($k$)'' for each $k=1,2,\ldots$.

As in the $k=0$ case, there are two steps to the proof.  First we will prove an ``in band'' lemma: 
\begin{lemma}[In-Band]\label{lem:k}Let $k\ge 1$ and assume that Condition C($k$) holds. Fix $\mu>0$ and $s<\nicefrac{1}{2}$. Then there is $\gamma_{k,\mu,s}>0$ so that if $\gamma<\gamma_{k,\mu,s}$ then there is $A_{k,\mu,s} <\infty$ such that
\begin{equation}\label{eq:inband}
	\Ev{\abs{G_\Lambda(x,\vm;y,\vxi;z)}^s} \ \le \ A_{k,\mu,s} \e^{-\mu \Upsilon_\Lambda(x,\vm;y,\vxi)}
\end{equation}
for any $\Lambda \subset \bbZ^d$,  $\ket{x,\vm},\ket{y,\vxi}\in \cE_\Lambda^{(k)}$, and  $z\in \bS_k$.
\end{lemma}
\noindent Second, we will prove the general results by applying the Combes-Thomas bound.

The proof of the in-band lemma relies on local decay inequalities, proved by a geometric decoupling argument, to  bound in-band correlations by restricting to regions with excitation number below $k$.  In \S\ref{sec:FMsteps} we present  two  local decay lemmas and show how they imply Lem.\ \ref{lem:k}.  In \S\ref{sec:filter} we describe the geometric decoupling argument and auxiliary bounds.  The local decay lemmas are proved in \S\ref{sec:localdecay}. Finally in \S\ref{sec:outofband} we complete the proof of Thm.\ \ref{thm:main} by proving the implication ``Condition C($k$) $\Rightarrow$ Theorem L($k$).''

\subsection{Local decay and in-band correlations}\label{sec:FMsteps}  Given oscillator states $\vm,\vxi\in\cM_\Lambda$, we define \be\label{eq:Gammamxix} \Gamma_{\vm|\vxi}(x)  \ := \  B_{R;\Lambda}(x), \quad R=R_{\vm|\vxi}(x).\ee 
When $x,y$ are distant enough so that  $\|x - y\|  > R_{\vm|\vxi}(x)=R$, the state $\vm$ has fewer than $k$ excitations in the interior of $\Gamma_{\vm|\vxi}(x)$, allowing us to prove the following
 \bl \label{lem:filtersmallR}  Let $k\ge 1$ and assume that Condition \emph{C($k$)} holds. Fix $\mu>0$ and $0<s<1$. Then there is a $\gamma_{k,\mu,s}^\sharp >0$ and $A^\sharp_{k,\mu,s} <\infty $ so that for $ \gamma <\gamma_{k,\mu,s}^\sharp$ we have 
 \be \label{fsR1}
          \bbE_{\Gamma}  \left ( |G_\Lambda(x,\vm; y,\vxi;z)|^s \right )
            \ \leq  \     A^\sharp_{k,\mu,s} \gamma^{s}
            \sum_{\ket{w,\vla} \in \cE_{ \Lambda \setminus \Gamma}^{(k)}}
             \e^{-\mu \norm{x-w}}     
            \abs{G_{\Lambda\setminus \Gamma}( w,\vla;y,\vxi_{\Lambda\setminus \Gamma}; z) }^s,
            \ee
 where $\Gamma = \Gamma_{\vm|\vxi}(x)$ and the bound holds
 for any $\Lambda \subset \bbZ^d$,  $\ket{ x,\vm } \in \cE_\Lambda^{(k)}$,  $\ket{y ,\vxi}  \in \cE_\Lambda$  with $y$  an exterior point $y\in  \Lambda \setminus \Gamma $, and $z\in \bS_k$.               \el
 
On the other hand, for $x,y$ sufficiently close, $\|x - y\| \le R_{\vm|\vxi}(x)$ and  the sites $x$ and $y$ are both contained in $\Gamma=\Gamma_{\vm|\vxi}(x)$.  Since $\vm$ and $\vxi$ differ outside of $\Gamma^\circ$, the states $\ket{x,\vm}$ and $\ket{y,\vxi}$ lie in orthogonal invariant subspaces for the Hamiltonian $H_{\Gamma^\circ \oplus \Lambda \setminus \Gamma^\circ}$,  leading to the following
\bl \label{lem:filter_large_R} 
Let $k\ge 1$ and assume that Condition \emph{C($k$)} holds. Fix $\mu>0$ and $0<s<\nicefrac{1}{2}$. Then there is $\wt{\gamma}_{k,\mu,s}>0$ and $\wt{A}_{k,\mu,s}<\infty$ so that for $\gamma <\wt{\gamma}_{k,\mu,s}$ we have 
\be\label{eq:filterlR} \bbE \left ( |G_\Lambda(x,\vm; y,\vxi; z)|^s \right ) \ \le \ \wt{A}_{k,\mu,s} \e^{-\mu R_{\vm|\vxi}(x)},
\ee
for any $\Lambda \subset \bbZ^d$,  $\ket{ x,\vm } \in \cE_\Lambda^{(k)}$,  $\ket{y ,\vxi}  \in \cE_\Lambda$, and $z\in \bS_k$.
  \el 

With these results we now give the
\begin{proof}[Proof of the in-band Lemma, Lem.\ \ref{lem:k}]
As a preliminary observation, note that it suffices to prove the result under the additional assumption that $\Lambda$ is finite, provided the bounds obtained are uniform in the size of $\Lambda$.  This is so because the Hamiltonians $H_{\Lambda_n}$ with $\Lambda_n=[-n,n]^D\cap \Lambda$ converge to $H_\Lambda$ in the strong resolvent sense. Second, by symmetry of the Green's function it suffices to prove eq.\ \eqref{eq:inband} with $\Upsilon_\Lambda$ replaced by the non-symmetric function
\be \ol{\Upsilon}_\Lambda(x,\vm;y,\vxi) \ = \ \max \left ( \|x-y\|, R_{\vm|\vxi}(x) \right ).\ee

Given a finite set $\Lambda \subset \bbZ^d$, we define 
 \be\label{eq:Fdef} F_\Lambda (k,\mu,s) \ = \ \max_{W\subset \Lambda} \sup_{z\in \bS_k}\max_{\ket{y,\vxi} \in \cE_W^{(k)} } \max_{x\in W}  \ \sum_{\vm \in\cM_W^{(k)} } \e^{\mu \ol{\Upsilon}_W(x,\vm;y,\vxi)}  \bbE_W(\abs{G_W(x,\vm;y,\vxi;z)}^s). \ee
 Since $\ol{\Upsilon}_\Lambda(x,\vm;y,\vxi) \le \diam \Lambda$ for any $\ket{x,\vm},\ket{y,\vxi} \in \cE_\Lambda$, we have the \emph{a priori} bound
 \be F_\Lambda(k,\mu,s) \ \le \ \frac{\kappa}{1-s} \e^{\mu \diam \Lambda} |\cM_\Lambda^{(k)}| \ < \ \infty,\ee
 by Lem.\ \ref{lem:apriori}.  Note that $\abs{\cM_\Lambda^{(k)}} = \abs{\Lambda}^k < \infty$.  
 
 To complete the proof, we must show that there is $\gamma_{k,\mu,s} <\infty$ such that for $\gamma <\gamma_{k,\mu,s}$ we have 
 \be F_\Lambda(k,\mu,s) \ \le \ C_{k,\mu,s}, \ee
 with a constant \emph{independent of $\Lambda$}.  We will obtain this bound by dividing the sum on the right hand side of eq.\ \eqref{eq:Fdef} into two parts, depending on whether $\ol{\Upsilon}_W=\|x-y\|$ or $R_{\vm|\vxi}(x)$.
 
Let $W\subset \Lambda$, $z\in \bS_k$, $\ket{y,\vxi}\in \cE_W^{(k)}$, and $x\in W$.  Suppose that  $\ol{\Upsilon}_W(x,\vm;y,\vxi) = \|x-y\| >  R_{\vm|\vxi}(x)= R$.  Let $\Gamma=\Gamma_{\vm|\vxi}(x)$. By Lem.\ \ref{lem:filtersmallR}, we have for $\gamma <\gamma^\sharp_{k,2\mu,s}$
\begin{multline}
	\e^{\mu\ol{\Upsilon}_W(x,\vm,y,\vxi)} \Ev{\abs{G_W(x,\vm;y,\vxi;z)}^s} \\
	\le \ A^\sharp_{k,2\mu,s} \gamma^s \sum_{\ket{w,\vla} \in \cE_{W\setminus \Gamma}^{(k)}} \e^{\mu\norm{x-y}-2\mu \norm{x-w}} \Ev{\abs{G_{W\setminus \Gamma}(w,\vla;y,\vxi;z)}^s}  \\
	\le \ A^\sharp_{k,2\mu,s} \gamma^s \sum_{\ket{w,\vla} \in \cE_{W\setminus \Gamma}^{(k)}} \e^{-\mu \norm{x-w}} \e^{\mu \norm{w-y}} \Ev{\abs{G_{W\setminus \Gamma}(w,\vla;y,\vxi;z)}^s} \\
	\le \ A^\sharp_{k,2\mu,s} \gamma^s \left ( \sum_{\twoline{w \in W }{ \|w-x\| >R}} \e^{-\mu \norm{w-x}} \right ) F_\Lambda(k,\mu,s) .
\end{multline}
Since the number of oscillator configurations $\vm\in \cM_\Lambda^{(k)}$ with $R_{\vm|\vxi}(x) \le  R$ is bounded by $\sigma_D^k (R\vee 1)^{Dk}$, we see that
\be 
\sum_{\twoline{ w\in W, \ \vm \in \cM_W^{(k)}}{R_{\vm|\vxi}(x) < \|x-y\|\wedge \|w -x\|} }\e^{-\mu \|w-x\|}   \ \le \ \sigma_D^k \sum_{w\in \bbZ^d} (\|w\|\vee 1)^{Dk} \e^{-\mu\|w\|} \ <\ \infty.
\ee
Thus 
\begin{equation}\label{eq:smallR}
	\sum_{\twoline{ \vm \in \cM_W^{(k)}}{R_{\vm|\vxi}(x) < \|x-y\|} }\e^{\mu\ol{\Upsilon}_W(x,\vm,y,\vxi)} \Ev{\abs{G_W(x,\vm;y,\vxi;z)}^s} \ \le \ B_{k,\mu,s} \gamma^s F_\Lambda(k,\mu,s),
\end{equation}
with $B_{k,\mu,s} < \infty$ independent of $\Lambda$.

On the other hand, if $\ol{\Upsilon}_W(x,\vm;y,\vxi) = R_{\vm|\vxi}(x) \ge \|x-y\|$, then an application of Lem.\ \ref{lem:filter_large_R} shows that
\be \e^{\mu \ol{\Upsilon}_W(x,\vm;y,\vxi)} \bbE \left ( \abs{G_W(x,\vm;y,\vxi;E+\im \epsilon)}^s \right ) \ \le \ \wt{A}_{k,2\mu,s} \e^{-\mu R_{\vm|\vxi}(x)} , \ee 
	provided $\gamma <\wt{\gamma}_{k,2\mu,s}$.  Thus
\begin{multline}\label{eq:largeR}
	\sum_{\twoline{ \vm \in \cM_W^{(k)}}{R_{\vm|\vxi}(x) \ge  \|x-y\|} }\e^{\mu\ol{\Upsilon}_W(x,\vm,y,\vxi)} \Ev{\abs{G_W(x,\vm;y,\vxi;z)}^s} \\
	\le \ \wt{A}_{k,2\mu,s} \sigma_D^k\sum_{R=\|x-y\|}^\infty (R\vee 1)^{Dk} \e^{-\mu R} \\
	\le \ \wt{A}_{k,2\mu,s} \sigma_D^k\sum_{R=0}^\infty (R\vee 1)^{Dk} \e^{-\mu R} \ =: \ b_{k,\mu,s} \ < \ \infty,
\end{multline}
where we have used again the fact that the number of oscillator configurations $\vm$ with $R_{\vm|\vxi}(x)\le R$ is bounded by $\sigma_D^k(R\vee 1)^{Dk} $.

Summing eqs.\ \eqref{eq:smallR} and \eqref{eq:largeR} and  maximizing over $W$, $\ket{y,\vxi}$, $x$ and $z$, we find that
\be \label{eq:closedform} F_\Lambda(k,\mu,s) \ \le \ b_{k,\mu,s} + B_{k,\mu,s}\gamma^s F_\Lambda(k,\mu,s),\ee
provided $\gamma < \gamma_{k,2\mu,s}^\sharp \wedge \wt{\gamma}_{k,2\mu,s}$. Thus, if
\be \gamma \ < \ \gamma_{k,\mu,s} \ := \ \min \left ( \gamma^\sharp_{k,\mu,s} , \ \wt{\gamma}_{k,\mu,s}, \ B_{k,\mu,s}^{-\frac{1}{s}} \right ) , \ee
we have \be F_\Lambda(k,\mu,s) \ \le \ \frac{b_{k,\mu,s}}{1- B_{k,\mu,s}\gamma^s} , \ee
independent of $\Lambda$.
\end{proof}

\subsection{Geometric Decoupling}\label{sec:filter} To prove the local decay lemmas, Lems.\ \ref{lem:filtersmallR} and \ref{lem:filter_large_R}, we will use a two step geometric resolvent identity  to decouple a ball from its complement.  A similar geometric decoupling argument was introduced in the context of random Schr\"odinger operators in \cite{Aizenman2001}.  In this section we adapt the argument of \cite{Aizenman2001} to the present context and prove auxiliary bounds.

For $\Gamma \subset \Lambda$ and $x \in \Gamma$, we  consider the following
     sum of the Green's function over boundary sites  of the subset $\Gamma$
    \be\label{Fterm}
            F_\Lambda^s (x,\vm; \Gamma;z) = 
        \sum_{w \in \partial \Gamma}
        \sum_{ 
                \twoline{\vla \in \cM_\Lambda}
                {  \vla_V = \vm_V   } } \abs{G_{\Gamma\oplus V}(x,\vm;w,\vla;z)}^s W_s(w,\vla) ,
      \ee
with $V=\Lambda\setminus \Gamma$ and
\be W_s(w,\vla) =  (\vla(w)\vee 1)^{\frac{1}{2}-\frac{s}{4}} \sum_{w'\sim w } (\vla(w')\vee 1)^{\frac{1}{2}-\frac{s}{4}} .\ee

A double application of the  geometric resolvent identities eqs.\ \eqref{gre}, \eqref{gre2} and  integration over the random variables corresponding to boundary points, yields the following `factorization' bound of the fractional moment. 
   \bp \label{prop:filter} For each $0< s <1$ there is $C_s <\infty$ such that if $\Gamma \subset \Lambda$, $V=\Lambda \setminus \Gamma$, $x\in \Gamma^\circ$ is an interior point of $\Gamma$, and $y\in V$, then \be 
    \bbE_{\partial \Gamma} \left ( |G_\Lambda (x,\vm; y,\vxi; z)|^s \right ) \leq   C_s \gamma^{2s} F_\Lambda^{s}(x,\vm;  \Gamma^\circ;z) F^{s}_\Lambda(y,\vxi;  V;z  ) .
               \ee
   \ep

\bpf By \eqref{gre},
\be G_\Lambda(x,\vm; y,\vxi;z)  = -\gamma \dirac{x,\vm}{ ( H_{\Gamma^\circ \oplus V^+ } - z)^{-1} 
                T_{\Gamma^\circ;V^+} 
      (H_{ \Lambda } - z)^{-1} }{y,\vxi },
               \ee
where $V^+=G \cup \partial \Gamma=\Lambda \setminus \Gamma^\circ$.
Applying \eqref{gre2} to the right hand side yields, 
\be   G_\Lambda(x,\vm; y,\vxi; z)  =  \gamma^2 \dirac{x, \vm}{( H_{\Gamma^\circ \oplus V^+ } - z)^{-1} 
                T_{\Gamma^\circ;V^+}       
        (H_{\Lambda } - z)^{-1}  
    T_{V;\Gamma} 
    (H_{\Gamma\oplus V } -z)^{-1}}{y,\vxi}.\ee
Introducing a partition of unity with respect to the basis $\cE_\Lambda$ on either side of the middle factor $T  (H - z)^{-1}T$ we obtain
\begin{multline} \label{eq:filterstep} G_\Lambda(x,\vm; y,\vxi; z) \\ =    
\gamma^2 \sum_{\twoline{x_1\in \partial \Gamma^\circ}{\vm_1 \in \cM_\Lambda}} \sum_{\twoline{ y_1\in \partial V}{\vxi_1\in \cM_\Lambda}} G_{\Gamma^\circ \oplus V^+} (x,\vm;x_1,\vm_1,z) \dirac{x_1,\vm_1}{T_{\Gamma^\circ;V^+}        
        (H_{\Lambda } - z)^{-1}  
    T_{V;\Gamma}}{y_1,\vxi_1} \\ \times G_{\Gamma\oplus V}(y_1,\vxi_1;y,\vxi,z).
               \end{multline}
Note that $\vm_1(u) =  \vm(u)$ for $u \notin \Gamma^\circ$ and $\vxi_1(u) = \vxi(u) $ for $u \not \in V$. 

The middle factor $\dirac{x_1,\vm_1}{T_{\Gamma^\circ;V^+}        
        (H_{\Lambda } - z)^{-1}  
    T_{V;\Gamma}}{y_1,\vxi_1}$ in each term of the sum on the right hand side of eq.\ \eqref{eq:filterstep} is the only factor that depends on the random variables $v_u$ for $u\in \partial \Gamma$.  Thus,
\begin{multline}    
    \bbE_{\partial \Gamma} \left ( |G_\Lambda (x,\vm; y,\vxi; z)|^s \right ) \\ \leq \ \gamma^{2s} 
    \sum_{\twoline{x_1\in \partial \Gamma^\circ}{\vm_1 \in \cM_\Lambda}} \sum_{\twoline{ y_1\in \partial V}{\vxi_1\in \cM_\Lambda}} \abs{G_{\Gamma^\circ\oplus V^+}(x,\vm;x_1,\vm_1,z)}^s \abs{G_{\Gamma\oplus V}(y_1,\vxi_1;y,\vxi,z)}^s \\
    \times \bbE_{\partial \Gamma} \left ( \abs{\dirac{x_1,\vm_1}{T_{\Gamma^\circ;V^+}        
        (H_{\Lambda } - z)^{-1}  
    T_{V;\Gamma}}{y_1,\vxi_1}}^s \right ) 
    \end{multline}  
Furthermore, by Lem.\ \ref{lem:apriori} and Prop.\ \ref{disp0},
\begin{multline} 
\bbE_{\partial \Gamma} \left ( \abs{\dirac{x_1,\vm_1}{T_{\Gamma^\circ;V^+}        
        (H_{\Lambda } - z)^{-1}  
    T_{V;\Gamma}}{y_1,\vxi_1}}^s \right ) \\ \le \
    \sum_{\twoline{x_2 \sim x_1}{\vm_2 \in \cM_\Lambda}} \sum_{\twoline{ y_2\sim y_1}{\vxi_1\in \cM_\Lambda}} \abs{K_{x_2;\vm_2}^{x_1,\vm_1}}^s \abs{K^{y_2;\vxi_2}_{y_1,\vxi_1}}^s \bb{E}_{\partial \Gamma} \left ( \abs{G_\Lambda(x_2,\vm_2;y_2,\vxi_2;z)}^s \right ) \\
    \le \ C_{s,\beta} \frac{\kappa^s}{1-s} W_s(x_1,\vm_1) W_s(y_1,\vxi_1),    \end{multline}
which completes the proof.        \epf
   
Note that the sum on the right hand side of eq.\ \eqref{Fterm} involves all oscillator states $\vla$ that agree with $\vm$ outside of $\Gamma$.  In particular, there are contributions from states with total excitation number ``off-shell'' for the band nearest to $z$.  However, we expect the bulk of the contribution to be due to states within the energy band.  In the following lemma, we  explicitly control the contribution of states outside the band $\bI_k$ to $F^s_\Lambda$. 
  This lemma does not rely on the induction hypothesis, Condition C($k$);  it only requires the Combes-Thomas bound Cor.\ \ref{cor:specCT}.
\bl \label{lem:ftermCT} Given $\nu >0$, $0<s<1$, and $k=0,1,\ldots$, there is $\wt{\gamma}_{k,\nu,s} >0$ such that if $\gamma <\wt{\gamma}_{k,\nu,s}$ then
\be \label{eq:ftermCT} F^{s}_\Lambda (x,\vm; \Gamma;z)    \ \le \ C_{\nu,k,s} \sum_{u\in \Gamma}\sum_{\twoline{\vla \in \cM_\Lambda^{(k)}}
                {  \vla_V = \vm_V } } \abs{G_{\Gamma\oplus V}(x,\vm;u,\vla;z)}^s 
        \sum_{w \in \partial \Gamma}
            \e^{-\nu  \|u - w\|} ,
           \ee
           where $V=\Lambda \setminus \Gamma$ and $z\in \bS_k$.
      \el
      
\bpf Note that $F_\Lambda^s (x,\vm; \Gamma;z)  \ = \ F_{\text{in}} + F_{\text{out}}$
where
\begin{multline}\label{eq:Finbound} F_{\text{in}} = \sum_{\twoline{w \in \partial \Gamma}{\vla \in \cM_\Lambda^{(k)}} } \abs{G_{\Gamma\oplus V}(x,\vm;w,\vla;z)}^s W_s(w,\vla)
                \\ \le \ 2D (k\vee 1)^{1-\frac{s}{2}} \sum_{\twoline{w \in \partial \Gamma}{\vla \in \cM_\Lambda^{(k)}} }\abs{G_{\Gamma\oplus V}(x,\vm;w,\vla;z)}^s,
\end{multline}
and
\be 
 F_{\text{out}} = \sum_{\twoline{u \in \partial \Gamma}{\vla' \not \in \cM_\Lambda^{(k)}} } \abs{G_{\Gamma\oplus V}(x,\vm;w,\vla';z)}^s W_s(w,\vla').
                 \ee

To bound $F_{\text{out}}$, we use the geometric resolvent identity eq.\ \eqref{firstres2}, which implies
\be
\left|  G_{\Gamma\oplus V} (x,\vm;u,\vla';z) \right| \ \le \ \gamma \abs{\dirac{x,\vm}{(H_{\Gamma \oplus V}- z)^{-1} T_\Gamma^{(k)}(H^{(k;\text{out})}_{\Gamma\oplus V} -z )^{-1}}{w,\vla'}} \ee 
for $\vla' \not \in \cM_\Lambda^{(k)}$. Introducing a partition of unity and applying Cor.\ \ref{cor:specCT} we obtain
\begin{multline}
	\left|  G_{\Gamma\oplus V} (x,\vm;w,\vla;z) \right| \\ \le \ C_{\nu,s}\gamma \sum_{\twoline{u\in \Gamma}{\vla\in \cM_\Lambda^{(k)}, \ \vla_V = \vla'_V}} \abs{G_{\Gamma\oplus V}(x,\vm;u,\vla;z)} (k\vee 1)^{\frac{1}{2}} \e^{-\frac{\nu}{s} d_\Gamma(u,\vla;w,\vla')},
\end{multline}
provided $\gamma \le \wt{\gamma}_{k,\nu,s}$. Here we have used the bound $Q(\vla,u)\le (2D+1) \sqrt{k\vee 1}$ for $\vla\in \cM_\Lambda^{(k)}$.  Thus
\begin{multline}\label{eq:Foutbound}
	F_{\text{out}} \ \le \ C_{k,\nu,s}  \sum_{\twoline{w \in \partial \Gamma}{\vla' \not \in \cM_\Lambda^{(k)}} } (N_\Lambda(\vla')\vee 1)^{1-\frac{s}{2}}\sum_{\twoline{u\in \Gamma}{\vla\in \cM_\Lambda^{(k)}}} \abs{G_{\Gamma\oplus V}(x,\vm;u,\vla;z)}^s \e^{-\nu d_\Gamma(u,\vla;w,\vla')} \\
	\le \ C_{k,\nu,s}  \sum_{\twoline{u\in \Gamma}{\vla\in \cM_\Lambda^{(k)}}} \abs{G_{\Gamma\oplus V}(x,\vm;u,\vla;z)}^s \sum_{w\in \partial \Lambda} \e^{-\nu \|u-w\|}
\end{multline}
by  Lem.\ \ref{lem:expsumbound2}.  Combining eq.\ \eqref{eq:Foutbound}  and eq.\ \eqref{eq:Finbound} completes the proof. \epf

The final lemma of this section provides a bound, depending on Condition C($k$), on the expectation of $F_\Lambda^s(x,\vm;\Gamma;z)$.
\bl \label{FM fterm} Assume Condition C($k$), let $0<s < 1$ and let $\mu >0$. Then there is $\gamma_{k,s,\mu} >0$ so that for $\gamma <\gamma_{k,s,\mu}$  and $\Lambda \subset \bbZ^d$ if for some $\ket{x,\vm} \in \cE_\Lambda^{(k)}$ and $R > 0$ it holds that
       \be \label{eq:localexsmall} \sum_{ u \in \Gamma} \vm(u) \ < \ k,\ee
       where $\Gamma = B_{R;\Lambda}(x)$,
       then
\be \label{f large R bound} 
         \bbE_{\Gamma} \left ( F^{s}_{\Lambda}
       (x,\vm; \Gamma;z)  \right )
       \leq C_{k,\mu,s} \e^{-\mu R} ,
                      \ee 
       for any $z \in \bS_k$.
       \el

       \bpf Fix $\nu >\mu$ and suppose that $\gamma<\wt{\gamma}_{k,\nu}$, with $\wt{\gamma}_{k,\nu}$ as in Lem.\ \ref{lem:ftermCT}.  Let $\ket{x,\vm}\in \cE_\Lambda^{(k)}$ satisfy eq.\ \eqref{eq:localexsmall}.  Applying the bound eq.\ \eqref{eq:ftermCT} to $F^s_\Lambda$ and then averaging with respect to $\bbE_{\Gamma}$ yields
       \be \bbE_{\Gamma} \left ( F_\Lambda^s(x,\vm;\Gamma;z) \right ) \ \le \ C_{\nu,k,s} \sum_{u \in \Gamma} \sum_{\vla \in \cM_\Lambda^{(k)} } \bbE_\Gamma \left ( \abs{G_{\Gamma\oplus V}(x,\vm;u,\vla;z)}^s \right ) \sum_{w\in \partial \Gamma} \e^{-\nu \norm{u-w}}. \ee
       In the invariant subspace containing $\ket{x,\vm}$ and $\ket{u,\vla}$,  $H_{\Gamma \oplus V}$ contains no hopping terms between states with differing oscillator excitations in $V$.  Thus $\vm_V = \vla_V$ and
       \begin{multline} G_{\Gamma\oplus V}(x,\vm;u,\vla;z) \ = \ \dirac{x,\vm_\Gamma}{(H_\Gamma + \omega N_V(\vm) - z)^{-1}}{u,\vla_\Gamma} \\ = \ G_\Gamma(x,\vm_\Gamma;u,\vla_\Gamma;z-\omega j),\end{multline}
       where $j=N_V(\vm)\ge 1$ by assumption. 
       
       Applying condition C($k$), we see that there is $\gamma^{\sharp}_{k,\nu,s} \le \wt{\gamma}_{k,\nu}$ so that for $\gamma <\gamma^{\sharp}_{k,\nu,s}$ we have 
       \be \bbE_\Gamma \left ( F_\Lambda^s(x,\vm;\Gamma;z) \right ) \ \le \ C_{\nu,k,s}   \sum_{u \in \Gamma}  \sum_{\vla \in \cM_\Gamma^{(k-j)} }  \sum_{w\in \partial \Gamma} \e^{-2\nu \Upsilon_\Gamma(x,\vm_\Gamma;u,\vla)-\nu \norm{u-w}}.\ee     
       By Lem.\ \ref{lem:expsumbound} we see that
       \begin{multline}  \bbE_\Gamma \left ( F_\Lambda^s(x,\vm;\Gamma;z) \right ) \ \le \ C_{\nu,k,s} \sum_{u\in \Gamma}\sum_{w\in \partial \Gamma} \e^{-\nu \|x-u\| -\nu \|u-w\|} \\ \le \ C_{\nu,k,s} R^{2D -1 } \e^{-\nu R} \ \le \ C_{\mu,\nu,k,s} \e^{-\mu R} .\qedhere\end{multline}
              \epf 

\subsection{Proof of the local decay lemmas}\label{sec:localdecay}
\subsubsection{Proof of Lem.\ \ref{lem:filtersmallR}}
 We consider the $R = 0$  and the $R > 0$ cases separately. The $R = 0$ case does not depend on the inductive hypothesis Condition C($k$), but only requires the Combes Thomas bound via Lemma \ref{lem:ftermCT}. It is covered by the following
  \bl\label{int step} Given $0<s<1$, $\nu >0$, and $k=0,1,\ldots$, there is $\wt{\gamma}_{k,\nu,s}$ such that if $\gamma <\wt{\gamma}_{k,\nu,s}$ then
  \be \bbE_{\{x\}} \left ( \abs{G_\Lambda(x,\vm;y,\vxi;z)}^s \right ) \ \le \ C_{k,\nu,s} \gamma^s \sum_{w\in \Lambda \setminus \{x\}} \sum_{\vla \in \cM_{\Lambda}^{(k)}} \e^{-\nu \norm{x-w}} \abs{G_{\{x\} \oplus\Lambda\setminus \{x\}}(w,\vla;y,\vxi;z)}^s \ee
 for any $\Lambda \subset \bbZ^d$,  $\ket{ x,\vm } \in \cE_\Lambda^{(k)}$,  $\ket{y ,\vxi}  \in \cE_\Lambda$  with $y\neq x$, and $z\in \bS_k$. 
  \el 
  
  \bpf The argument runs parallel the beginning of the proof of the $k=0$ in-band lemma, Lem.\ \ref{lem:k=0}.  By eq.\ \eqref{gre2} we have  
\be  G_\Lambda(x,\vm;y,\vxi ; z) \ = \  
                - \gamma \dirac{ x,\vm}{
                    (H_{\Lambda} - z)^{-1}
              T_{\{x\};\Lambda\setminus \{x\}}
              (H_{\{x\}\oplus \Lambda\bks\{x\} } - z)^{-1 }}{y,\vxi}.\ee
Introducing a partition of unity to the left of $T_{\{x\};\Lambda\setminus\{x\}}$, we obtain
 \begin{multline}    G_\Lambda(x,\vm;y,\vxi ; z) \\ = \ - \gamma
               \sum_{\vm_1\in \cM_\Lambda }   \dirac{ x,\vm}{
                    (H_{\Lambda} - z)^{-1} }{x, \vm_1}
                    \dirac{ x,\vm_1 }{T_{\{x\} ; \Lambda\setminus \{x\}}
     (H_{\{x\}; \Lambda\bks\{x\} } - z)^{-1 }} {y,\vxi}.  \end{multline}
Since the second factor in the sum is independent of $v_x$, it follows from Lem.\ \ref{lem:apriori} that
\begin{multline} \bbE_{\{x\}} \left ( |G_\Lambda(x,\vm;y,\vxi ; z)|^s \right ) \ \le \  \frac{\kappa}{1-s} \gamma^s \sum_{\vm_1\in \cM_\Lambda }   
                    \abs{\dirac{ x,\vm_1 }{T_{\{x\} ; \Lambda\setminus \{x\}}
     (H_{\{x\}; \Lambda\bks\{x\} } - z)^{-1 }} {y,\vxi}}^s \\
    \le \ C_{k,s} \gamma^s F^s_\Lambda(y,\vxi; \Lambda \setminus \{x\}).
    \end{multline}
    An application of Lemma \ref{lem:ftermCT} completes the proof.
 \epf
 
 When $R>0$, we apply  Prop.\ \ref{prop:filter} followed by averaging with $\bbE_{\Gamma^\circ}$ to conclude that
  \be \bbE_{\Gamma}  \left ( |G(x,\vm; y,\vxi; z)|^s  \right ) \ \le \
  C_s \gamma^{2s} \bbE_\Gamma \left ( F_\Lambda^{s}(x,\vm;  \Gamma^\circ ;z) \right ) F^{s}_\Lambda(y,\vxi;  V;z  ),
               \ee
  where $V=\Lambda \setminus \Gamma$.
  We apply Lem.\ \ref{FM fterm} to the first factor and  Lem.\ \ref{lem:ftermCT} to the second factor, taking $\nu >\mu$.  It follows, for $\gamma$ sufficiently small, that
  \begin{multline} \bbE_{\Gamma} \left ( |G(x,\vm; y,\vxi; z)|^s  \right ) \\ \le \ C_{\mu,\nu,k,s} \gamma^s \e^{-\mu R_{\vm|\vxi}(x)} \sum_{u\in V}\sum_{w \in \partial V}
            \e^{-\nu  \|u - w\|}\sum_{\vla \in \cM_\Lambda^{(k)} } \abs{G_{\Gamma\oplus V}(u,\vla,;y,\vxi;z)}^s  \\
        \le \ C_{\mu,\nu,k,s} \gamma^s \sum_{u\in V}   \e^{-\mu \|x-u\|}\sum_{\vla \in \cM_\Lambda^{(k)} } \abs{G_{\Gamma\oplus V}(u,\vla,;y,\vxi;z)}^s. \qedhere
         \end{multline}
  This completes the proof of Lem.\ \ref{lem:filtersmallR}.

\subsubsection{Proof of Lem.\ \ref{lem:filter_large_R}.} If $R=R_{\vm|\vxi}(x)=0$, then the result follows from the \emph{a-priori} bound, Lem.\ \ref{lem:apriori}. Suppose that $R>0$ and let $\Gamma = \Gamma_{\vm|\vxi}(x)$. Then, we have that $\sum_{u \in  \Gamma^\circ} |\vm(u)|= k-\ol{k}$ with $\ol{k} \ge 1$. Furthermore, $G_{\Gamma^\circ\oplus \Lambda \setminus\Gamma^\circ}(x,\vm;y,\vxi;z) =0$,
since $x\in \Gamma^\circ$ and $\vm$ and $\vxi$ differ on $\Lambda \setminus \Gamma^\circ$. Thus by eq.\ \eqref{gre},
\be G_{\Lambda}(x,\vm; y,\vxi; z) \ = \
- \gamma  \dirac{x,\vm}{ ( H_{\Gamma^\circ \oplus \Lambda \setminus\Gamma^\circ  } - z)^{-1} 
T_{\Gamma^\circ \oplus \Lambda \setminus\Gamma^\circ }      
             (H_{\Lambda} - z)^{-1}}{y,\vxi}
                   \ee
As $s<\nicefrac{1}{2}$, we have 
\begin{multline}
\bbE (|G_\Lambda (x,\vm; y,\vxi; z)|^s ) \\ \leq \ \gamma^s
             \sum_{ y_1,  \vxi_1} \left[   \bbE \left ( 
             	\abs{\dirac{x,\vm}{( H_{\Gamma^\circ\oplus \Lambda \setminus\Gamma^\circ } - z)^{-1}
               T_{ \Gamma^\circ \oplus \Lambda \setminus \Gamma^\circ }}{y_1,  \vxi_1}}^{2s}
               \right)  \right]^{\frac{ 1}{2} }  \\ \times 
              \left[  \bbE \left ( \abs{\dirac{y_1, \vxi_1}{  (H_{\Lambda} - z)^{-1}}{y,\vxi}}^{2s}\right ) \right]^{ \frac{1}{2}}  
\end{multline}
by the Cauchy-Shwarz inequality.
We apply the \emph{a-priori bound}, Lem.\ \ref{lem:apriori}, to the second factor and then expand to obtain
\begin{multline}
	\bbE (|G_\Lambda (x,\vm; y,\vxi; z)|^s ) \\ \leq \ \left [ \frac{\kappa}{1-2s} \right ]^{\frac{1}{2}} \gamma^s
             \sum_{ y_1,  \vxi_1} \left[   \bbE \left ( 
             	\abs{\dirac{x,\vm}{( H_{\Gamma^\circ\oplus \Lambda \setminus\Gamma^\circ } - z)^{-1}
               T_{ \Gamma^\circ \oplus \Lambda \setminus \Gamma^\circ }}{y_1,  \vxi_1}}^{2s}
               \right)  \right]^{\frac{ 1}{2} } \\
               \le  \ C_{s} \gamma^s \sum_{\twoline{u\in \partial \Gamma^\circ}{y_1\sim u}} \sum_{\vla,\vxi_1}\abs{K_{u,\vla}^{y_1;\vxi_1}}^s  \left[   \bbE\left ( 
             	\abs{\dirac{x,\vm}{( H_{\Gamma^\circ\oplus \Lambda \setminus\Gamma^\circ } - z)^{-1}}{u,  \vla}}^{2s}
               \right)  \right]^{\frac{ 1}{2} } \\
               \le C_{s} (k\vee 1)^{1-\frac{s}{2}} \gamma^s \sum_{u\in \partial \Gamma^\circ} \sum_{\vla}  \left[   \bbE\left ( 
             	\abs{\dirac{x,\vm}{( H_{\Gamma^\circ\oplus \Lambda \setminus\Gamma^\circ } - z)^{-1}}{u,  \vla}}^{2s}
               \right)  \right]^{\frac{ 1}{2} },
\end{multline}
where we have used Prop.\ \ref{disp0} in the final step. Since $x\in \Gamma^\circ$, the terms in the sum on the right hand side vanish unless $\vla_{\Lambda\setminus\Gamma^\circ} = \vm_{\Lambda\setminus \Gamma^\circ}$.  Thus
\begin{multline}
\sum_{u\in \partial \Gamma^\circ} \sum_{\vla}  \left[   \bbE\left ( 
             	\abs{\dirac{x,\vm}{( H_{\Gamma^\circ\oplus \Lambda \setminus\Gamma^\circ } - z)^{-1}}{u,  \vla}}^{2s}
               \right)  \right]^{\frac{ 1}{2} } \\ = \  \sum_{u\in \partial \Gamma^\circ} \sum_{\vla \in \cM_{\Gamma_0}}  \left[   \bbE\left ( 
             	\abs{\dirac{x,\vm_{\Gamma^\circ}}{( H_{\Gamma} +\omega \ol{k}-z)^{-1}}{u,  \vla}}^{2s}
               \right)  \right]^{\frac{ 1}{2} },
\end{multline} 
where the term $\omega \overline{k}$ accounts for the energy of $\vm$ in $\Lambda \setminus \Gamma^\circ$. By the induction hypothesis, Condition C($k$), we see that there is $\wt{\gamma}_{k,\mu,s}$ so that for $\gamma<\wt{\gamma}_{k,\mu,s}$ we have
\begin{multline}
	\bbE (|G_\Lambda (x,\vm; y,\vxi; z)|^s )  \ \leq \ C_{k,s} \sum_{u\in \partial \Gamma^\circ}\sum_{\lambda \in \cM_{\Gamma^\circ}} \e^{-3\mu \Upsilon_{\Gamma^\circ}(x,\vm_{\Gamma^\circ};u,\vla) -3\mu\cR^{(k-\ol{k})}_{\Gamma^\circ}(\vm_{\Gamma^\circ},\vla)} \\
	\le \ C_{k,s} \sum_{u\in \partial \Gamma^\circ}\sum_{\lambda \in \cM_{\Gamma^\circ}} \e^{-3\mu \Upsilon_{\Gamma^\circ}(x,\vm_{\Gamma^\circ};u,\vla) -3\mu \abs{\sqrt{N_{\Gamma^\circ}(\lambda)} -\sqrt{k-\ol{k}}}} \\
	\le \ C_{\mu,k,s} \sum_{u\in \partial \Gamma^\circ} \e^{-2\mu \norm{x-u}} \ \le \ C_{\mu,k,s} \e^{-\mu R},
	\end{multline}
	where we have used Lem.\ \ref{lem:expsumbound}.
	This completes the proof of Lem.\ \ref{lem:filter_large_R}.

 \subsection{Out of band correlations: the proof of Thm.\ L($k$)}\label{sec:outofband} First note that it suffices to prove ``Condition C($k$) $\Rightarrow$ Theorem L($k$)'' for $s<\nicefrac{1}{2}$.  Indeed, the $\nicefrac{1}{2}\le s<1$ case of Theorem L($k$)  is a consequence of the result for $s <\nicefrac{1}{2}$ by the all-for-one lemma, Lem.\ \ref{lem:allforone}.   
 
 If $\ket{x,\vm}$ and  $\ket{y,\vxi}$ are in $\cE_\Lambda^{(k)}$ the result follows directly from the in-band lemma, Lem.\ \ref{lem:k}. If $\ket{x,\vm} \in \cE_\Lambda^{(k)}$ and $\ket{y,\vxi} \not \in \cE_\Lambda^{(k)}$, we see that
     \begin{equation}
     	\abs{G_\Lambda(x,\vm;y,\vxi;z)}^s  \ \le \ \gamma^s \sum_{\ket{w,\vla} \in \cE^{(k)}_\Lambda}  \abs{G_\Lambda(x,\vm;w,\vla;z)}^s
     \abs{\dirac{w,\vla}{T_{\Lambda}^{(k)} \frac{1}{H_\Lambda^{(k)} -z }}{y,\vxi}}^s,
     \end{equation} 
     by eq.\ (\ref{firstres}).
     By the Combes-Thomas bound, Cor.\ \ref{cor:specCT}, there is $\alpha_{k,\mu,s}$ such that if $\gamma <\alpha_{k,\mu,s}$ then 
     \begin{multline}
     	\abs{G_\Lambda(x,\vm;y,\vxi;z)}^s  \ \le \ C_{k,\mu,s}  \sum_{\ket{w,\vla} \in \cE^{(k)}_\Lambda} \abs{G_\Lambda(x,\vm;w,\vla;z)}^s \e^{-2\mu d_\Lambda(w,\vla;y,\vxi)} \\
     	\le \ C_{k,\mu,s}  \sum_{\ket{w,\vla} \in \cE^{(k)}_\Lambda} \abs{G_\Lambda(x,\vm;w,\vla;z)}^s \e^{-2\mu (\Upsilon_\Lambda(w,\vla;y,\vxi) + \cR_\Lambda^{(k)}(\vla;\vxi) )}.
     \end{multline}
          Averaging and applying the in-band lemma, Lem.\ \ref{lem:k}, we obtain
     \begin{multline}
     \Ev{\abs{G_\Lambda(x,\vm;y,\vxi;z)}^s} \ \le \ C_{k,\mu,s} \sum_{\ket{w,\vla} \in \cE^{(k)}_\Lambda} \e^{-\mu \Upsilon_\Lambda(x,\vm;w,\vla;z)} \e^{-2\mu (\Upsilon_\Lambda(w,\vla;y,\vxi) + \cR_\Lambda^{(k)}(\vla;\vxi) )} \\
     \le \ C_{k,\mu,s} \e^{-\mu ( \Upsilon_\Lambda(x,\vm;y,\vxi;z)+\cR_\Lambda^{(k)}(\vm;\vxi))} \sum_{\ket{w,\vla} \in \cE^{(k)}_\Lambda}  \e^{-\mu (\Upsilon_\Lambda(w,\vla;y,\vxi) + \cR_\Lambda^{(k)}(\vla;\vxi) )} \\
     \le \ C_{k,\mu,s} \e^{-\mu ( \Upsilon_\Lambda(x,\vm;y,\vxi;z)+\cR_\Lambda^{(k)}(\vm;\vxi))} ,
     \end{multline}
     where we have applied Lem.\ \ref{lem:expsumbound} in the last step. 
     If $\ket{y,\vxi}\in \cE_\Lambda^{(k)}$ and $\ket{x,\vm}\not \in \cE_\Lambda^{(k)}$, the argument is similar, but involves the reversed geometric resolvent identity eq.\ \eqref{firstres2}.    If both $\ket{x,\vm}$ and $\ket{y,\vxi}$ are in $\cE_\Lambda \setminus \cE_{\Lambda}^{(k)}$ a similar argument invoking both eqs.\ \eqref{firstres} and \eqref{firstres2} implies the result.  This completes the proof of Thm.\ L($k$).  
\appendix

  \section{Fractional moment tools}\label{sec:fmt}  In this section we prove the \emph{a priori} bound Lem.\ \ref{lem:apriori}.  A preliminary observation is that the fractional moment bound eq.\ \eqref{eq:fmb} is a consequence of a  weak-$L^1$ type estimate.  Given $x,y\in \Lambda$, let $\Pr_{\{x,y\}}$ denote the conditional probability, conditioned on $\Sigma_{\{x,y\}^c}$ \tem i.e., given the random variables  on $\Omega_{\Lambda\setminus \{x,y\}}$.
  \bl\label{lem:wkL1} There is a $C <\infty$ such that for all $z\in \bbC$
  \be \label{eq:wkL1} \Pr_{\{x,y\} }\left [ \abs{G_\Lambda(x,\vm;y,\vxi;z)} > t \right ] \ \le \ \frac{C}{t} \ee
  \el 
  
 Lem.\ \ref{lem:apriori} follows easily from  Lem.\ \ref{lem:wkL1}. Indeed, assuming eq.\ \eqref{eq:wkL1} we have 
  \begin{multline}
  	\bbE_{\{x,y\}} \left ( \abs{G_\Lambda(x,\vm;y,\vxi;z)}^s \right ) \ = \ \int_0^\infty  \Pr_{\{x,y\} }\left [ \abs{G_\Lambda(x,\vm;y,\vxi;z)}^s > t \right ] \di t \\ \le \ \int_0^{C^s}  \di t + C \int_{C^s}^\infty  \frac{1}{t^{\frac{1}{s}}} \di t \
  	\ = \ \frac{C^s }{1-s}.
  \end{multline}
  
The weak-$L^1$ bound eq.\ \eqref{eq:wkL1} is a consequence of an  operator inequality which played a key role in the extension of the moment method to continuum Schr\"odinger operators \cite{Aizenman2006}:
\begin{lemma}{\cite[Lemma 3.1]{Aizenman2006}}\label{dissipative} 
       There is a universal constant $C_W<\infty $ such that
        \be   \abs{\setb{ v \in \bbR }{ \| B_1 (H - z - v)^{-1}B_2 \|   > t }}\ \leq \  C_W \, \|B_1\|_{HS } \|B_2\|_{HS} \,   \frac{1}{t}   \ee
        whenever $H$ is a self adjoint operator on a Hilbert space $\fh$, $z\in \bbC \setminus \bbR$, and  $B_1$, $B_2$ are  Hilbert-Schmidt operators on $\fh$.
    \end{lemma}
    \noindent As explained in \cite{Aizenman2006}, the following is an immediate corollary of the above result:
    \bl{\cite[Proposition 3.2]{Aizenman2006}}\label{dissipative schrodinger}
    With $C_W$ as in Lem.\ \ref{dissipative},
        \begin{multline}
        	\abs{\setb{  (v_1,v_2)\in[0,1]^2  }{  \| B_1 U_1^{1/2}(H  +v_1 U_1  + v_2 U_2 - z )^{-1}U_2^{1/2} B_2 \|   > t }}
                   \\ \leq \  2 C_W \|B_1\|_{HS } \|B_2\|_{HS}  \frac{1}{t}   
        \end{multline}   
  whenever $U_1$, $U_2$ are nonnegative operators on a Hilbert space $\fh$, $H$ is a self adjoint operator on $\fh$, $z\in \bbC \setminus \bbR$, and  $B_1$, $B_2$ are  Hilbert-Schmidt operators on $\fh$.
   \el
   
   We now prove Lem.\ \ref{lem:wkL1} from these two Lemmas.  The key observation is that 
   \be \abs{G_\Lambda(x,\vm;y,\vxi;z)} \ = \ \norm{ B_{x,\vm} P_{\cS_x} \frac{1}{H_\Lambda - z} P_{\cS_y} B_{y,\vm}},\ee
   where $B_{x,\vm}$ denotes the rank one \tem hence, Hilbert-Schmidt\tem operator $B_{x,\vm}  =\ket{x,\vm}\bra{x,\vm}$, $P_{\cS_x}$ denotes the projection onto the span of $\cS_x = \setb{\ket{x,\vla}}{\vla \in \cM_\Lambda}$, and similarly for $B_{y,\vxi}$ and $P_{\cS_y}$. If $x\neq y$, Lem.\ \ref{dissipative schrodinger} implies that 
   \begin{multline}
   	\Pr_{\{x,y\}} \left [ \abs{G_\Lambda(x,\vm;y,\vxi;z)} > t \right ]
   	\\ \le \ \norm{\rho}_\infty^2 \abs{\setb{(v_x,v_y) \in [0,V_+]^2}{\norm{ B_{x,\vm} P_{\cS_x} \frac{1}{H_\Lambda - z} P_{\cS_y} B_{y,\vm}}>t}} \\
   	\le \ 2 C_W \norm{\rho}_\infty^2 V_+^2 \norm{B_{x,\vm}}\norm{B_{y,\vm}} \frac{1}{t}  \ = \  2 C_W V_+^2\norm{\rho}_\infty^2 \frac{1}{t} ,
   \end{multline}
   since $\norm{B_{x,\vm}}=\norm{B_{y,\vm}}=1$. 
   If $x=y$, a similar argument based directly on Lemma \ref{dissipative}  yields
   \begin{equation}
   \Pr_{\{x\}} \left [ \abs{G_\Lambda(x,\vm;x,\vxi;z)} > t \right ]	\ \le \ C_W \norm{\rho}_\infty V_+ \frac{1}{t}.
   \end{equation}
   Thus Lemma \ref{lem:wkL1} holds as claimed

\section{Displacement states}\label{displacement}

Our goal in this section is to derive bounds on the matrix elements  of the displacement operators $D(\beta) = \e^{\beta b^\dagger - \beta^* b}$.  The key result is the following
\begin{proposition}
	For any $\beta \in \bbC$, $\mu \in \bbR$ and $n\in \mathbb{N}$,
	\begin{equation}\label{eq:squaresumid} \sum_{m\in \N}\e^{2\mu (m-n)} \abs{\dirac{m}{D(\beta)}{n}  }^2 \ = \ \e^{\left (\e^{2\mu}-1\right ) |\beta|^2}  L_n\left (-|\beta|^2 \left ( \e^{\mu}-\e^{-\mu}\right)^2 \right )	\ ,
	\end{equation}
	where $L_n$ is the Laguerre polynomial of order $n$,
	\begin{equation}\label{eq:Laguerre}
		L_n(x) \ = \ \sum_{k=0}^n \frac{1}{k!} {n \choose k} (-x)^k.
	\end{equation}
\end{proposition}
\begin{proof}  Let $N=b^\dagger b$ denote the number operator and let $\mu \in \bbR$.  Then 
 \be 
	\e^{\mu N} b^\dagger \e^{-\mu N} = \e^{\mu} b^\dagger \quad \text{and} \quad \e^{\mu N} b \e^{-\mu N}= \e^{-\mu} b,
	\ee
and thus  
\be \e^{\mu N} D(\beta)\e^{-\mu N} \ = \ \e^{-\frac{1}{2}|\beta|^2} \e^{\e^\mu \beta b^\dagger} \e^{-\e^{-\mu} \beta^* b}\ = \ \e^{\frac{1}{2}\left (\e^{2\mu}-1\right ) |\beta|^2} D(\e^{\mu}\beta)
	\e^{\left (\e^{\mu}-\e^{-\mu}\right )\beta^* b}. \ee
	Therefore 
\begin{equation}\e^{\mu (m-n)} \dirac{m}{D(\beta)}{n}  \ = \   \e^{\frac{1}{2}\left (\e^{2\mu}-1\right ) |\beta|^2}
\dirac{m}{D(\e^{\mu}\beta)
	\e^{\left (\e^{\mu}-\e^{-\mu}\right )\beta b}}{n}.\label{eq:id1}
	\end{equation}
	Squaring both sides and summing over $n$ we obtain
	\be \sum_{m}\e^{2\mu (m-n)} \abs{\dirac{m}{D(\beta)}{n}  }^2 \ = \ \e^{\left (\e^{2\mu}-1\right ) |\beta|^2} \norm{\e^{\left (\e^{\mu}-\e^{-\mu}\right )\beta b}\ket{n}}^2 ,\ee
	since $D(\e^\mu \beta)$ is unitary. 
	
	Recall that 
\be b \ket{k} \ = \ \begin{cases}\sqrt{k}\ket{k-1} & \text{if } k\ge 1,\\
 0 & \text{if } k = 0.	
 \end{cases}
\ee
Therefore, for any $\gamma\in \bbC$,
\be \e^{\gamma b} \ket{n} \ = \ \sum_{k=0}^n \frac{1}{k!} \sqrt{\frac{n!}{(n-k)!}} \gamma^{k} \ket{n-k},\ee
and thus
\be \norm{ \e^{\gamma b} \ket{n}}^2 \ = \ \sum_{k=0}^n \frac{1}{k!} {n \choose k} \abs{\gamma}^{2k}\ = \ L_n(-\abs{\gamma}^2) \ . \qedhere
\ee
\end{proof}

\begin{corollary} Let $\beta \in \bbC$ and $\mu \in \R$.  Then,
	\begin{equation}
	\lim_{n\rightarrow \infty} \sum_{m \in \N} \e^{\frac{2\mu}{\sqrt{n}}(m-n)}	\abs{\dirac{m}{D(\beta)}{n}  }^2 \ = \ I_0(4 \abs{\beta}\mu)
	\end{equation}
	where $I_0$ denotes the modified Bessel function of the first kind. In particular
\begin{equation}\label{eq:squaresumbound}
	\sup_{n\in \N} \sum_{m\in \N} \e^{\frac{2\mu}{\sqrt{n\vee 1}}(m-n)}	\abs{\dirac{m}{D(\beta)}{n}  }^2  \ < \ \infty \ .
	\end{equation}
\end{corollary}
\begin{proof} Replacing $\mu$ by $\nicefrac{\mu}{\sqrt{n}}$ in eq.\ \eqref{eq:squaresumid} yields
\be  \sum_{n\in \N} \e^{\frac{2\mu}{\sqrt{n}}(m-n)} \abs{\dirac{m}{D(\beta)}{n}  }^2 \ = \ \e^{ (\e^{\frac{2\mu}{\sqrt{n}}}-1 ) |\beta|^2}  L_n\left (-|\beta|^2 \left ( \e^{\frac{\mu}{\sqrt{n}}}-\e^{-\frac{\mu}{\sqrt{n}}}\right)^2 \right ),\ee 
yielding 
\begin{multline}\lim_{n\rightarrow \infty} \sum_{n\in \N} \e^{\frac{2\mu}{\sqrt{n}}(m-n)}	\abs{\dirac{m}{D(\beta)}{n}  }^2 \ = \ \lim_{n\rightarrow \infty} L_n\left (-\frac{4|\beta|^2\mu^2}{n} \right )\\ =  \ \lim_{n\rightarrow \infty} \sum_{k=0}^n \frac{1}{(k!)^2} \frac{n (n-1)\cdots (n-k+1)}{n^k} \left ( 4 |\beta|^2 \mu^2 \right )^k \\ = \ \sum_{k=0}^\infty \frac{1}{(k!)^2}  \left ( 4 |\beta|^2 \mu^2 \right )^k \ = \ I_0(4  \abs{\beta}\mu) \ . \qedhere  
\end{multline}
\end{proof}

It is useful to rephrase these bounds in terms of the following metric on $\bbN$:
\begin{equation}\label{eq:appmetric} d(m,n) \ = \ \abs{\sqrt{m}-\sqrt{n}}   .	
\end{equation}
The above estimates imply the following

\begin{theorem}\label{thm:squaresum}For any $\mu>0$,	\be  \sup_{n\in \N} \sum_{m\in \N} \e^{\mu d(m,n)} \abs{\dirac{m}{D(\beta)}{n}}^2 <\infty\ .\ee
\end{theorem}
	
\begin{proof} 
Since $d(n,m)=d(m,n)$, 
\begin{multline} \sup_{n} \sum_{m} \e^{\mu d(m,n)} \abs{\dirac{m}{D(\beta)}{n}}^2 \\ \le \ 1
+ \sup_{n} \sum_{m<n} \e^{\mu (\sqrt{n}-\sqrt{m})} \abs{\dirac{m}{D(\beta)}{n}}^2 + \sup_n \sum_{m>n}  \e^{\mu (\sqrt{m}-\sqrt{n})}
 	\abs{\dirac{m}{D(\beta)}{n}}^2 \\ < \  
 	\sup_{n} \sum_{m} \e^{\frac{\mu}{\sqrt{n}}(n-m)} \abs{\dirac{m}{D(\beta)}{n}}^2 + \sup_{n} \sum_{m} \e^{\frac{\mu}{\sqrt{n}} (m-n)}
 	\abs{\dirac{m}{D(\beta)}{n}}^2 \ <\ \infty. \qedhere 
 	\end{multline}
\end{proof}

Since $D(\beta)$ is unitary, we have $\abs{\dirac{m}{D(\beta)}{n}} \le 1$ and thus
\be  \sup_{n\in \N} \sum_{m\in \N} \e^{\mu d(m,n)} \abs{\dirac{m}{D(\beta)}{n}}^p <\infty\ \ee
for any $p \ge 2$.  To bound such sums with $p< 2$, we  use H\"older's inequality and the following 
\begin{lemma}\label{lem:expsum}For any $\mu >0$ and $\alpha \in \bbR$,
	\be  \sum_{m\in \N} (m\vee 1)^\alpha \e^{-\mu d(m,n)} \ \le  C_{\mu,\alpha} (n\vee 1)^{\frac{1}{2}+\alpha}.\ee
\end{lemma}
\begin{proof}  Let
\be S_\alpha(n) \ := \ \sum_{m=0}^\infty (m\vee 1)^{\alpha} \e^{-\mu d(m,n)} \ = \ \e^{-\mu\sqrt{n}} + \sum_{m=1}^\infty m^\alpha \e^{-\mu \abs{\sqrt{m}-\sqrt{n}}}.\ee
Note that $S_\alpha(n) < \infty$ for any $\alpha$ and $n$.   Thus it suffices to show that
\be \limsup_{n\rightarrow \infty} n^{-\alpha-\frac{1}{2}} S_\alpha(n) \ < \ \infty. \ee
In particular, we may assume without loss of generality that $n \ge 1$. Fix $n\ge 1$ and let $m=n+k$ where $k\ge -n$.  For $k \ge  0$, we have
\be
d(n+k,n) \ = \ \sqrt{n+k} -\sqrt{n} \ = \ \frac{1}{2} \int_{n}^{n+k} \frac{1}{\sqrt{t}}\di t \ \ge \  \frac{1}{2}\frac{k}{\sqrt{n+k}},
\ee
and for $-n \le k < 0$, we have
\be 
d(n+k,n) \ = \ \sqrt{n} - \sqrt{n+k} \ \ge \ \frac{1}{2} \frac{\abs{k}}{\sqrt{n}} \ \ge \ \frac{1}{2} \frac{\abs{k}}{\sqrt{n+|k|}}. \ee
Thus for any $k \ge -n$ we have
\be
d(n+k,n) \ \ge \  \frac{1}{2\sqrt{2}} \begin{cases} \frac{\abs{k}}{\sqrt{n}} & \text{if } \abs{k}\le n ,\\
 \sqrt{k} & \text{if } k > n .	
 \end{cases}
\ee
Thus
\begin{multline}
	S_\alpha(n) \ \le \  \sum_{m=0}^{2n}(m\vee 1)^\alpha \e^{-\frac{\mu}{2\sqrt{2n}} \abs{m-n}} + \sum_{m=2n+1}^\infty m^\alpha \e^{-\frac{\mu}{2\sqrt{2}} \sqrt{m-n}} \\
	\le \ \sum_{k=-n}^n (n+\abs{k})^\alpha \e^{-\frac{\mu}{2\sqrt{2n}} \abs{k}} + \sum_{k=n+1}^\infty k^\alpha\e^{-\frac{\mu}{2\sqrt{2}} \sqrt{k}}.
\end{multline} 
The first term is  $O(n^{\alpha + \frac{1}{2}})$, while the second  $O(\e^{-\delta \sqrt{n}})$ for $\delta >0$.
 \end{proof}
 
\bc For any $\mu>0$ and $0 < p < 2$ there is $C_{\mu,p}<\infty$ such that
\be   \sum_{m\in\N} \e^{\mu d(m,n)} \abs{\dirac{m}{D(\beta)}{n}}^p < C_{\mu,p}  \left ( n \vee 1 \right )^{\frac{1}{2}-\frac{p}{4}} \ .\ee 
\ec
\bpf
By H\"older's inequality we have
\begin{multline}
	\sum_{m} \e^{\mu d(m,n)} \abs{\dirac{m}{D(\beta)}{n}}^p
\\ \le \ \left ( \sum_{m} \e^{\frac{2}{p}(\mu+1) d(m,n)} \abs{\dirac{m}{D(\beta)}{n}}^2 
\right )^{\frac{p}{2}} \left ( \sum_{m} \e^{-\frac{2}{2-p} d(m,n)} \right )^{\frac{2-p}{2}}.
\end{multline} 
The right hand side is bounded as claimed by Theorem \ref{thm:squaresum} and Lemma \ref{lem:expsum}.
\epf

\section{Combes-Thomas bound}
 \label{CT section}In this section we prove a Combes-Thomas bound for the disordered Holstein Hamiltonian, Thm.\ \ref{thm:CT} below. To begin we derive a version of the Combes-Thomas bound in an abstract setting. 
 
 \subsection{An abstract Combes-Thomas bound} Let $\mc{G}$ be a countable set and $d$ a metric on $\mc{G}$. Our goal in this section is to obtain an estimate of the form
 \begin{equation}\label{eq:CT} \abs{\dirac{n}{\frac{1}{H-z}}{m}} \ \le \ A_z \e^{-\mu_z d(n,m)}	
 \end{equation}
 for the matrix elements of the resolvent of a self-adjoint operator $H$.  Throughout we assume that the set of functions with finite support, 
 \begin{equation}\label{eq:finitesupport} \mathcal{F} \ = \ \setb{ \phi\in \ell^2(\cG)}{\phi(n)=0 \text{ for $n$ outside a finite set}}
\end{equation}
is contained in the domain of $H$.  For $\psi\in \mc{F}$ we have the $H \psi(n)  =  \sum_{m}\dirac{n}{H}{m} \psi(m)$, 
where $\dirac{n}{H}{m} \ = \ \ipc{\delta_n}{H\delta_m}$ is the operator kernel of $H$. Our key assumption is that this kernel decays exponentially away from the diagonal in the sense that
 \begin{equation}\label{eq:Malpha} M_\alpha \ := \ \sup_{n\in \mc{G}} \left ( \sum_{m\neq n} \abs{\dirac{n}{H}{m}}^2  \e^{\alpha d(n,m)} \right )^{\frac{1}{2}}  \ < \ \infty \ 
 \end{equation}
 for some $\alpha >0.$  
 
 More generally, we may consider subsets $\cS,\cT\subset \cG$ and the corresponding projections $P_{\cS}$, $P_{\cT}$ onto the subspaces spanned by $\setb{\ket{n}}{n\in \cS}$ and $\setb{\ket{n}}{n\in \cT}$ respectively.  We shall also obtain estimates of the form
 \be\label{eq:PCT}  \norm{P_{\cS} \frac{1}{H-z} P_{\cT}} \ \le \ A_z \e^{-\mu_z d(\cS,\cT)},\ee 
  where 
 \be d(\cS,\cT) = \inf_{\substack{n\in \cS , \ m\in \cT}} d(n,m).\ee
 
An estimate of the form eq.\ \eqref{eq:CT} or eq.\ \eqref{eq:PCT} is called a Combes-Thomas bound. Such estimates are well known to follow from a bound of the form
 \be \label{inf1}  \sup_{n\in \mc{G}} \sum_{m\neq n} \abs{\dirac{n}{H}{m}} \e^{\alpha d(n,m)} \ < \ \infty
 \ee
 (see \cite[Theorem 10.5]{aizenman2015random}).  If the metric is uniformly exponentially summable, i.e.,
 \be \sup_m \sum_{n} \e^{-\mu d(n,m)} \ < \ \infty\ee 
 for some $\mu>0$,  then eq.\ \eqref{inf1} \tem\ and thus eq.\ \eqref{eq:CT} \tem \ follows from eq.\ \eqref{eq:Malpha} and the Cauchy-Schwarz inequality.   Our goal here is to obtain eq.\ \eqref{eq:CT} from eq.\ \eqref{eq:Malpha} under an assumption weaker that uniform exponential summability. 
 
For each $m \in \mc{G}$, let
\be  \mc{G}_m \ = \ \setb{n\in \mc{G}}{\dirac{n}{H}{m}\neq 0}.\ee  In place of uniform exponential summability, we shall require that,  for some $\mu <\alpha$,
\begin{equation}\label{eq:Fmu} F_\mu(m) \ := \ \sqrt{\sum_{n  \in \mc{G}_m } \e^{-\mu d(n,m)}}	
\end{equation}
is finite for every $m\in \mc{G}$. However, we will not assume $F_\mu$ to be uniformly bounded.

In general, eq.\ (\ref{eq:Malpha}) may not guarantee a unique self-adjoint extension of the operator $H$ restricted to $\mc{F}$.  However, the Holstein Hamiltonians considered in this paper are easily seen to be essentially self-adjoint on $\mc{F}$, since the off-diagonal parts of these operators are bounded. Thus it suffices for our purposes to restrict our attention to  operators that are essentially self-adjoint on $\mc{F}$. 
\begin{theorem}\label{thm:abstCT} Let $H$ be a self-adjoint operator on $\ell^2(\mc{G})$ with domain $\mc{D} \supset \mc{F}$ and suppose that $\mc{F}$ is a core for $H$.  Suppose that eq.\ \eqref{eq:Malpha} holds for the operator kernel of $H$ for some $\alpha >0$.  If for some $\mu < \alpha$, $F_\mu$ is relatively $H$ bounded \tem  i.e., there are $a,b\ge 0$ such that
\begin{equation}\norm{F_\mu \psi} \ \le \ a \norm{H\psi} + b \norm{\psi}\label{eq:Fmurelbound}	
\end{equation}
for all $\psi \in \mc{F}$ \tem  then for each $z\in \bbC \setminus \sigma(H)$  and $\cS,\cT \subset \mc{G}$
\be  \norm{P_{\cS} \frac{1}{H-z}P_{\cT}} \ \le \ A_{z;\nu} \e^{-\nu d(\cS,\cT)}\ee
where $\nu$ is any number such that
\be 0 \ \le \ \nu \ < \ \frac{\alpha-\mu}{2M_\alpha} \min \left \{M_\alpha,\frac{1}{ a},  \frac{1}{ a \norm{\frac{H}{H-z}} + b \norm{\frac{1}{H-z}} } \right \} ,\ee
and
\be  A_{z;\nu} \ = \ \frac{\norm{\frac{1}{H-z}}}{1- \frac{2 M_\alpha \nu}{\alpha -\mu} \left (  a \norm{\frac{H}{H-z}} + b \norm{\frac{1}{H-z}} \right )},\ee
which is finite.
\end{theorem}

\begin{proof}Fix $\cT \subset G$ and $L>0$. Let $h_{L;\cT}(n)  =  d(n,\cT)\wedge L.$
Thus $h_{L;\cT}(n) \le L$ for all $n$. We wish to show that 
$H_\nu = \e^{\nu h_{L;\cT}} H \e^{-\nu h_{L;\cT}}$
is a small perturbation of $H$ and use this to bound the matrix elements of the resolvent of $H$.   Initially we define $H_\nu$ on $\mc{F}$.  Given $\psi\in \mc{F}$, we have $ \e^{-\nu h_{L;\cT}} \psi\in \mc{F}\subset \mc{D}$ and thus it makes sense to take 
\begin{equation}\label{eq:Hepsdefn} H_\nu\psi \ := \ \e^{\nu h_{L;\cT}} H \e^{-\nu h_{L;\cT}}\psi .
\end{equation}
Let
\be\label{eq:deltaHeps}\delta H_\nu \psi \ :=\ H_\nu \psi - H \psi.
\ee

If $p>\nu$, then  \begin{multline}
	\abs{\dirac{n}{\delta H_\nu}{m}} \ = \ \abs{ \e^{\nu (h_{L;\cT}(n) -h_{L;\cT}(m))} - 1  } \abs{\dirac{n}{H}{m}} \\ \le \ \abs{\e^{\nu d(n,m)} -1} \abs{\dirac{n}{H}{m}} \ \le \ 
	\abs{\e^{\nu d(n,m)} -1}\e^{-p d(n,m)} \e^{p d(n,m)} \abs{\dirac{n}{H}{m}} \\ \le \ \frac{  \nu}{p}     \e^{p d(n,m)} I[n\neq m] \abs{\dirac{n}{H}{m}},
\end{multline} 
since $\max_{x> 0} (e^{\nu x}-1)\e^{-p x} \le \nicefrac{ \nu}{p}.$
Taking $p = \nicefrac{\alpha-\mu}{2}$ we find
\begin{multline}\label{eq:deltaHepsrelbound}
	\norm{\delta H_\nu \psi}^2 \ = \ \sum_n \abs{\sum_{m} \dirac{n}{\delta H_\nu}{m} \psi(m)}^2 \ = \ \sum_n \abs{\sum_{m} \dirac{n}{\delta H_\nu}{m}\e^{\frac{\mu}{2}d(n,m)} \e^{-\frac{\mu}{2} d(n,m)} \psi(m)}^2 \\
	  \le \ \frac{4\nu^2}{(\alpha-\mu)^2} \left [ \sup_{n} \sum_{m\neq n} \abs{\dirac{n}{H}{m}}^2 \e^{\alpha d(n,m)} \right ] \sum_{n}\sum_{m\in G_n} \e^{-\mu d(n,m)} \abs{\psi(m)}^2 \\ = \ \frac{4 M_\alpha^2}{(\alpha-\mu)^2} \nu^2 \norm{F_\mu \psi}^2 \  \le \ \frac{4 M_\alpha^2}{(\alpha-\mu)^2} \nu^2  \left ( a \norm{H\psi} + b \norm{\psi}\right )^2.
\end{multline}

Now for any $\psi \in \mc{H}$,
$ \norm{\e^{-\nu h_{L;\cT}} \psi }  \le \norm{\psi}.$
Thus for $\psi \in \mc{F}$,
\be  \norm{H \e^{-\nu h_{L;\cT}}\psi} \ \le \  \norm{H_\nu \psi} \ \le 
\ \norm{H\psi} + \norm{\delta H_\nu \psi} \ \le \  C_{\nu} \left ( \norm{H\psi} + \norm{\psi} \right )\ee
where $C_\nu <\infty$ by eq.\ \eqref{eq:deltaHepsrelbound}. Since $\mc{F}$ is a core for $H$ it follows by taking limits that $\e^{-\nu h_{L;\cT}}\psi \in \mc{D}$ whenever $\psi \in \mc{D}$.  Hence the definitions eqs.\ (\ref{eq:Hepsdefn}, \ref{eq:deltaHeps}) and the bound eq.\ \eqref{eq:deltaHepsrelbound} can be extended to $\psi\in \mc{D}$.  Furthermore if  
\be \nu < \frac{\alpha-\mu}{2 M_\alpha a}\ee 
it follows from eq.\ \eqref{eq:deltaHepsrelbound} that
$\delta H_\nu=H_\nu-H$ is $H$-bounded with relative bound $<1$ and thus by \cite[Ch 4, Thm. 1.1]{kato1995perturbation}  that $H_\nu$ is a closed operator on $\mc{D}$.

Since $H_\nu=\e^{\nu h_{L;\cT}}H\e^{-\nu h_{L;\cT}}$, with $\e^{\pm \nu h_{L;\cT}}$ bounded,  we have 
\be \frac{1}{H-z} \ = \ \e^{-\nu h_{L;\cT}} \frac{1}{H_\nu-z} \e^{\nu h_{L;\cT}}\ee
for any $z$ in the resolvent set of either $H$ or $H_\nu$; thus $\sigma(H) =\sigma(H_\nu)$.  Given $z$ in the common resolvent set, we conclude from eq.\ \eqref{eq:deltaHepsrelbound} that
\be \norm{ \delta H_\nu \frac{1}{H-z}} \ \le \ \frac{2M_\alpha }{\alpha -\mu} \nu \left ( a \norm{H\frac{1}{H-z}} + b \norm{\frac{1}{H-z}} \right ). \ee
Thus the Neumann series
\be  \frac{1}{H_\nu - z}  \ = \ \sum_{n=0}^\infty \frac{1}{H-z} \left [ \delta H_\nu \frac{1}{H-z} \right ]^n\ee
converges provided 
\be  \nu < \frac{\alpha -\mu}{2 M_\alpha \left ( a \norm{\frac{H}{H-z}} + b \norm{\frac{1}{H-z}} \right ) }.\ee
For such $\nu$ and $z$ we have, since $h_{L;\cT}(m)=0$ for $m\in \cT$, 
\begin{multline}
\norm{P_{\cS} \frac{1}{H-z} P_{\cT}} \ = \ \norm{\e^{-\nu h_{L;\cT}}P_{\cS} \frac{1}{H_\nu -z}P_{\cT}} \ \le \ \e^{-\nu (d(\cS,\cT)\wedge L)}\norm{\frac{1}{H_\nu -z}} \\
	\le \ \frac{\norm{\frac{1}{H-z}}}{1- \frac{2M_\alpha }{\alpha -\mu} \nu \left ( a \norm{H\frac{1}{H-z}} + b \norm{\frac{1}{H-z}} \right ) } \e^{-\nu (d(\cS,\cT)\wedge L)}.
\end{multline}
Taking $L\rightarrow \infty$ yields the desired result.\end{proof}

 \subsection{Combes-Thomas estimates for the Disordered Holstein Hamiltonian} In this section we apply Thm.\ \ref{thm:abstCT} to the operator $H_{\cS}(\gamma)$, seen as an operator on $\ell^2(\cS)$, where $\cS$ is any subset of  $\cE_\Lambda$ with $\Lambda \subset \bbZ^d$. Recall the definition of the metric $d_\Lambda$ on $\cE_\Lambda$ \tem see eq.\ \eqref{eq:dLambda}.
\begin{theorem}\label{thm:CT} Fix $\mu >0$ and let
\be S_\mu \ := \ 2\sqrt[4]{2D}\e^{-2\mu } \frac{ \mu^3}{ (1+\mu)^2} \frac{1}{  \sup_{n} \sum_{m} \abs{\dirac{n}{D(\beta)}{m}}^2 \e^{4\mu \abs{\sqrt{n}-\sqrt{m}}}}.\ee
Then $S_\mu>0$ and given $\Lambda\subset \bbZ^D$,  $\mc{S}\subset \cE_\Lambda$,  $z \in \bbC \setminus \sigma(H_{\mc{S}}(\gamma))$, and $\cS_1$, $\cS_2 \subset \mc{S}$, we have
\begin{equation}\label{eq:CTbound} \norm{P_{\cS_1}\frac{1}{H_{\mc{S}}(\gamma)-z} P_{\cS_2}} \ \le \ A_{z;\nu} \e^{-\nu d_\Lambda(\cS_1;\cS_2)}	
\end{equation}
where
\be 0 \ \le \ \nu \ < \  \, \min \left \{\mu,   \frac{ S_\mu\omega}{\gamma }  ,    \frac{ S_\mu\omega}{\gamma } \frac{1}{ \left (  \norm{\frac{H_{\mc{S}}}{H_{\mc{S}}-z}} + 4\omega  \norm{\frac{1}{H_{\mc{S}}-z}} \right ) } \right \}, \ee 
and
\be  A_{z;\nu} \ =  \ \frac{\norm{\frac{1}{H_{\mc{S}}-z}}}{1- \nu\frac{\gamma}{S_\mu\omega} \left (   \norm{\frac{H_{\mc{S}}}{H_{\mc{S}}-z}} + 4\omega \norm{\frac{1}{H_{\mc{S}}-z}} \right )},\ee
with $H_{\mc{S}}=H_{\mc{S}}(\gamma)$.
\end{theorem}
\begin{remark}
	In the applications of this bound in the present work, we take $z = E1 + \im \epsilon$ with $|\epsilon|<1$ and $\dist(E,\sigma(H_{\mc{S}}))=\Delta >0$, in which case we may use the estimates
	\be  \norm{\frac{H_{\mc{S}}}{H_{\mc{S}}-z}} \ \le \ 1 + \frac{|E|}{\Delta} \quad \text{and} \quad \norm{\frac{1}{H_{\mc{S}}-z}} \ \le \ \frac{1}{\Delta}.\ee 
	Thus the bound \eqref{eq:CTbound} holds with 
	\be \nu < \min \left \{ \mu, S_\mu \frac{\omega}{\gamma}  \frac{\Delta}{\Delta + 2|E| + 4\omega } \right \}\ee
	and
	\be  A_{z;\nu} \ \le \ \frac{1}{\Delta - \nu \frac{\gamma}{S_\mu\omega}\left (   \Delta + 2|E| + 4\omega\right )}.\ee 
	Choosing $\nu \le  \frac{S_\mu}{2} \frac{\omega}{\gamma}  \frac{\Delta}{\Delta + 2|E| + 4 \omega }$, we have $A_{z;\nu} \le \nicefrac{2}{\Delta}$.
\end{remark}

\begin{proof}To begin, note that $S_\mu >0$ by Thm.\ \ref{thm:squaresum}.  Furthermore
\begin{align}M_{\alpha}(\mc{S}) \ :=& \  \left ( \sup_{\ket{x,\vm}\in \mc{S}} \sum_{\ket{y,\vec{\xi}}\neq \ket{x,\vm} \in \cE_\Lambda} \abs{\dirac{x,\vm}{H_{\Lambda}}{y,\vxi}}^2 \e^{\alpha d_\Lambda(x,\vec{m};y,\vec{\xi})} \right )^{\frac{1}{2}}
\\ \nonumber \le& \ \gamma  \left ( \sup_{\ket{x,\vm}\in \cE_\Lambda} \sum_{\substack{ y\in \Lambda , \ y\sim x \\ \vxi\in \cM_\Lambda } }
\abs{K_{x,\vm}^{y,\vxi}}^2 \e^{\alpha d_\Lambda(x,\vec{m};y,\vec{\xi})} \right )^{\frac{1}{2}}\\ \nonumber 
=& \ \gamma \left ( \sup_{\ket{x,\vm}\in \cE_\Lambda} \sum_{\substack{y\in \Lambda , \, y\sim x  \\ \xi_x,\, \xi_y \in \bbN } }
	\abs{ \dirac{ \bm{m}(x) }{ D(-\beta) }{\xi_x}}^2
	\abs{\dirac{ \bm{m}(y) }{ D(\beta) }{ \xi_y}}^2 \e^{\alpha \left (1 + \abs{\sqrt{\vm(x)}-\sqrt{\xi_x}} + \abs{\sqrt{\vm(y)} - \sqrt{\xi_y}} \right )} \right )^{\frac{1}{2}} \\
	\nonumber 
	\le& \ \gamma \sqrt{2D} \e^{\nicefrac{\alpha}{2}}  \sup_{n} \sum_{m} \abs{\dirac{n}{D(\beta)}{m}}^2 \e^{\alpha \abs{\sqrt{n}-\sqrt{m}}}  \ = \   (2D)^{\nicefrac{3}{4}}\frac{\gamma}{2}  \frac{\alpha^3}{(4+\alpha)^2}  \frac{1}{S_{4\alpha}}. 
\end{align}	

Let $\mu >0$ and let
\be F_\mu(x,\vec{m};\mc{S}) \ := \ \sqrt{ \sum_{\ket{y,\vxi}\in \mc{S}}I \left [ \dirac{y,\vxi}{H_{\mc{S}}(\gamma)}{x,\vm} \neq 0 \right ] \e^{-\mu d(x,\vm;y,\vxi)} }.\ee 
Then
\begin{multline}F_\mu(x,\vec{m};\mc{S})^2 \ \le \   \sum_{\substack{y\sim x \\ n,m \in \bbN}}\e^{-\mu (\abs{\sqrt{\vm(x)+n}-\sqrt{\vm(x)}} + \abs{\sqrt{\vm(y)+n}-\sqrt{\vm(y)}} )} \\ \le \   \left ( \frac{\mu + 2}{\mu} \right )^4 \sqrt{\vm(x) \vee 1} \sum_{y\sim x}\sqrt{\vm(y) \vee 1},
\end{multline}
by Lem.\ \ref{lem:expsum}. Thus for any $\psi \in P_{\mc{S}}\mc{H}_\Lambda$ we have
\begin{align}
\norm{F_\mu(x,\vec{m})\psi}^2 \ \le & \ \left ( 1 + \frac{2}{\mu} \right )^4 \frac{1}{2} \sum_{\ket{x,\vec{m}}\in \mc{S}} \sum_{y\sim x} \left ( 
\sum_{y\sim x} \frac{1}{\sqrt{2D}} \vm(x)\vee 1 + \sqrt{2D} \vm(y)\vee 1 \right ) \abs{\psi(x,\vec{m})}^2 \\ \nonumber 
\le & \ 	\left ( 1 + \frac{2}{\mu} \right )^4 \left ( \sqrt{2D}\norm{\psi}^2  + \frac{\sqrt{2D}}{2}\sum_{\ket{x,\vec{m}}\in \mc{S}} \sum_{y\in \Lambda}  \vm(y)  \abs{\psi(x,\vec{m})}^2 \right ) \\ \nonumber 
\le &  \ \left ( \frac{\mu+2}{\mu} \right )^4 \sqrt{2D} \left (\frac{1}{2\omega} \ipc{\psi}{H_{\mc{S}}(\gamma)\psi} + \norm{\psi}^2  \right ) \\\nonumber 
\le & \ \left ( \frac{\mu+2}{\mu} \right )^4 \sqrt{2D} \left ( \frac{t^2}{4\omega^2}  \norm{H_{\mc{S}}(\gamma) \psi}^2 + \left ( 1 +\frac{1}{2t^2} \right ) \norm{\psi}^2 \right ),
\end{align}
where $t>0$ is arbitrary.
Taking square roots and $t=1$ we find the upper bound  
\begin{equation}
\norm{F_\mu(x,\vec{m})\psi} \ \le \ \left ( \frac{\mu+2}{\mu} \right )^2 \sqrt[4]{2D} \left ( \frac{1}{2 \omega}  \norm{H_{\mc{S}}(\gamma)\psi} + 2 \norm{\psi} \right ).
\end{equation}

Thus the hypotheses of Thm. \ref{thm:abstCT} are satisfied with $\mc{G}=\mc{S}$, $H = H_{\mc{S}}(\gamma)$, 
\be  a= \left ( \frac{\mu +2}{\mu} \right )^2 \sqrt[4]{2D} \frac{1}{2 \omega} \ ,\quad \text{and} \quad b=  2 \left ( \frac{\mu +2}{\mu} \right )^2 \sqrt[4]{2D}  .\ee
Taking $\mu = \nicefrac{\alpha}{2}$ yields the result.
\end{proof}

\section{Dynamical Localization}\label{sec:DL}
In this section, we sketch a proof of Thm.\ \ref{thm:gf2ec}, following closely arguments in ref.~\cite{AW}.  The core of the argument is an abstract averaging principle which may be formulated as follows.  We fix a separable Hilbert space $\cH$ and a self-adjoint operator $H_0$ on $\cH$ with compact resolvent.  Let $P$ be a projection on $\cH$ and consider the family \be\label{eq:Hv} H_v=H_0+v P, \quad v\in \bbR. \ee  For each $v\in \bbR$, $H_v$ is a self-adjoint operator with compact resolvent.  Thus given $\phi,\psi\in \cH$ and an interval $I\subset \bbR$ it makes sense to define the correlator
\be\label{eq:Qv} Q_v(\phi,\psi;I) \ = \ \sum_{E\in \sigma(H_v)\cap I} \abs{\ipc{\phi}{P_{E}(H_v)\psi}}.\ee
Note that 
\be
Q_v(\phi,\phi) \ = \ \sum_{E \in \sigma(H_v)\cap I} \ipc{\phi}{P_{E}(H_v) \phi} \ = \ \norm{P_I(H_v)\phi}^2,
\ee
and
\be
\abs{Q_v(\phi,\psi;I)} \ \le \ Q_v(\phi,\phi;I)^{\frac{1}{2}} Q_v(\psi,\psi;I)^{\frac{1}{2}} \ = \  \norm{P_I(H)\phi}\norm{P_I(H)\psi},\ee 
by Cauchy-Schwarz applied to the Hilbert space inner product and the sum over $E$.

Thm.\ \ref{thm:gf2ec} follows  from the following abstract averaging bound
\bt\label{thm:aba} For any $0 < s <1$, any interval $I\subset \bbR$, any $\phi \in \ran P$ with $\norm{\phi}=1$, any $\psi\in \cH$,   any $m\in \bbN$, and any $v\in \bbR$,
\be \int_{\bbR} Q_{u+v}(\phi,\psi;I)^{2-s} \frac{\di u}{|u|^s} \ \le \ \norm{P_I(H)\psi}^{2-2s}  \sum_{n=1}^N  \int_{I} \abs{\ipc{\phi_n}{(H_v - E)^{-1}\psi}}^s \di E,\ee
where $\{\phi_n\}_{n=1}^N$, $N=\dim \ran P$, is any orthonormal basis for the range of $P$.
\et
\begin{remark}
	The value $N=\infty$ is allowed.
\end{remark}

Before sketching the proof of Thm.\ \ref{thm:aba}, let us show how it may be used to prove Thm.\ \ref{thm:gf2ec}. To begin, note that given  positive $h:[0,V_+]\rightarrow [0,\infty)$ we have
\begin{multline}\int_0^{V_+} h(v) \rho(v) \di v \ = \ \int_0^{V_+}\int_0^{V_+} h(v-v'+v') \rho(v)\di v\rho(v')\di v' \\ \le \ \frac{\norm{\rho}_\infty}{V^s_+} \int_0^{V_+} \left ( \int_{-\infty}^\infty h(v+u) \frac{\di u}{|u|^s} \right ) \rho(v) \di v.	
\end{multline}
Thus
\begin{multline}\bbE_{\{x\}}\left ( Q_\Lambda(x,\vm;y,\vxi;\bJ_k)^{2-s} \right ) \ = \  \int_{0}^{V_+} Q_\Lambda(x,\vm;y,\vxi;\bJ_k)^{2-s} \rho(v_x)\di v_x   \\ \le \ \frac{\norm{\rho}_\infty}{V^s_+}
\int_0^{V_+} \left ( \int_{-\infty}^\infty Q_{\Lambda;v_x \mapsto v+u} (x,\vm;y,\vxi;\bJ_k)^{2-s} \frac{\di u}{|u|^s} \right ) \rho(v) \di v.
\end{multline}
Applying Thm.\ \ref{thm:aba} yields
\begin{multline}\label{eq:intermeddl}\bbE_{\{x\}}\left ( Q_\Lambda(x,\vm;y,\vxi;\bJ_k)^{2-s} \right ) 
\\ \le \  \frac{\norm{\rho}_\infty}{V^s_+} \norm{P_{\bJ_k}\ket{y,\vxi}}^{2-2s} \sum_{\vla}  \int_0^{V_+}\int_{0}^{\omega k+V_++4D\gamma} \abs{\dirac{x,\vla}{\frac{1}{H_\Lambda - E  }}{y,\vxi}}^s \di E \rho(v_x) \di v_x ,
\end{multline}
where $P_{\bJ_k}=P_{\bJ_k}(H_\Lambda)$ is the spectral projection for $H_\Lambda$ onto $\bJ_k = [0,\omega k + V_++4D\gamma]$.
By H\"older's inequality
\be 
\Ev{Q_\Lambda(x,\vm;y,\vxi;\bJ_k)} \ \le \  \Ev{Q_\Lambda(x,\vm;y,\vxi;\bJ_k)^{2-s}}^{\frac{1}{2-s}}.
\ee 
Thus, averaging eq.\ \eqref{eq:intermeddl}, we obtain
\begin{multline} \Ev{Q_\Lambda(x,\vm;y,\vxi;\bJ_k)} \\ \le \ \norm{P_{\bJ_k} \ket{y,\vxi}}^{1 - \frac{s}{2-s}}
\left [  \frac{\norm{\rho}_\infty}{V^s_+} \sum_{\vla} \int_{0}^{\omega k + V_+ + 4D\gamma} \bbE_{\{x\}} \left ( \abs{G_\Lambda(x,\vla;y,\vxi;E)}^s\right ) \di E \right ]^{\frac{1}{2-s}}.\end{multline}
Since
\begin{multline}
\norm{P_{\bJ_k}(H)\ket{y,\vxi}}^2 \ = \ \dirac{y,\vxi}{P_{\bJ_k}(H)}{y,\vxi} \\ \le \ (1+\omega(k+1)) \dirac{y,\vxi}{\frac{1}{H+1}}{y,\vxi} \ \le \ \frac{1 + \omega(k+1)}{1+ \omega N_\Lambda(\vxi)} \ \le \ C_{\omega} \frac{k\vee 1}{N_\Lambda(\vxi)\vee 1},
\end{multline}
we find that
\begin{multline}
	\bbE_{\{x\}}\left ( Q_\Lambda(x,\vm;y,\vxi;\bJ_k)^{2-s} \right ) 
	\\ \le C    \left [ \frac{k \vee 1}{N_\Lambda(\vxi) \vee 1} \right ]^{1-\frac{s}{2-s}}\sum_{\vla} \int_{0}^{\omega k+V_+ + 4D\gamma} \bbE_{\{x\}} \left ( \abs{G_\Lambda(x,\vla;y,\vxi;E)}^s\right ) \di E,
\end{multline}
which is eq. \eqref{eq:gf2ec}.  This completes the proof of Thm.\ \ref{thm:gf2ec}, assuming Thm.\ \ref{thm:aba}.

Thm.\ \ref{thm:aba} follows from an identity for a family of correlators that interpolate between $Q_v(\phi,\psi;I)$ and $Q_v(\phi,\phi;I)$:
\begin{equation}\label{eq:Qs}
Q_v(\phi,\psi;I;s) \ = \ \sum_{E\in \sigma(H_v)\cap I}  \abs{\ipc{\phi}{P_{E}(H_v)\phi}}^{1-s} \abs{\ipc{\phi}{P_{E}(H_v)\psi}}^s. 
\end{equation}
Since 
\be \abs{\ipc{\phi}{P_{E}(H_v)\psi}} \ \le \ \sqrt{ \ipc{\phi}{P_{E}(H_v)\phi} \ipc{\psi}{P_{E}(H_v)\psi}},\ee
$Q_v(\phi,\psi;I;s)$ is well defined in the range $0 \le s \le 2$ and satisfies
\begin{multline} Q_v(\phi,\psi;I;s) \ \le \ \sum_{E\in \sigma(H_v)\cap I} \abs{\ipc{\phi}{P_{E}(H_v)\phi}}^{1-\frac{s}{2}} \abs{\ipc{\psi}{P_{E}(H_v)\psi}}^{\frac{s}{2}} \\ \le \ Q_v(\phi,\phi;I)^{1-\frac{s}{2}} Q_v(\psi,\psi;I)^{\frac{s}{2}} \\  = \ \norm{P_I(H)\phi}^{2-s} \norm{P_I(H)\psi}^s \ = \  \norm{\phi}^{2-s} \norm{\psi}^s,	
\end{multline}
by H\"older's inequality. Furthermore, as a function of $s$, $Q_v(\cdot;s)$ is log-convex.  Thus
\begin{multline}\label{eq:s=1bound} Q_v(\phi,\psi;I,1) \ \le \ Q_v(\phi,\psi;I,2)^{\frac{1-s}{2-s}} Q_v(\phi,\psi;I,s)^{\frac{1}{2-s}}  \\ \le \  \norm{P_I(H)\psi}^{2 \frac{1-s}{2-s}} Q_v(\phi,\psi;I,s)^{\frac{1}{2-s}}.\end{multline}

Regarding the correlator $Q_v(\phi,\psi;I,s)$ we have the following identity, which adapts Lemma 4.4 of \cite{AW} to the present context:
\begin{lemma}\label{lem:aba}
	Let $0 <s <1$.  Then for all $\phi\in \ran P$ and $\psi \in \cH$ we have
	\begin{multline}\label{eq:sbound}
	\int_{\bbR}Q_{v+u}(\phi,\psi;I,s) \frac{\di u}{|u|^s} \\ = \ \int_{I} \sum_{\kappa \in \sigma(K_v(E))} \ipc{\phi}{\Pi_\kappa(E)\phi}^{1-s} \abs{\ipc{\phi}{\Pi_\kappa(E) P \left (H_v -E \right )^{-1} \psi}}^{s}  \di E,
	\end{multline}
	where $\Pi_\kappa(E)$ is the spectral projection onto the eigenspace for eigenvalue $\kappa$ for the $E$-dependent compact operator
	\begin{equation}
		K_v(E) \ = \ P \left (H_v - E \right )^{-1} P,
	\end{equation}
	which acts on the range of $P$.
\end{lemma}
For a proof of Lem.\ \ref{lem:aba}, we refer to the proof of Lemma 4.4 of \cite{AW}.  Although, the result and proof given in ref.\ \cite{AW} are formulated in the specific context of multi-particle Schr\"odinger operators, they are easily generalized to the abstract setting given here with only notational changes.

To prove Thm.\ \ref{thm:aba}, we put eqs.\ \eqref{eq:s=1bound} and \eqref{eq:sbound} together to obtain
\begin{multline}\int_{\bbR}Q_{v+u}(\phi,\psi;I)^{2-s} \frac{\di u}{|u|^s} \\ \le \ 
 \norm{P_I(H)\psi}^{2-2s}
\int_{I} \sum_{\kappa \in \sigma(K_v(E))} \ipc{\phi}{\Pi_\kappa(E)\phi}^{1-s} \abs{\ipc{\phi}{\Pi_\kappa(E) P \left (H_v -E \right )^{-1} \psi}}^{s}  \di E.\end{multline}
Introducing the partition of unity  $P= \sum_{n}\ipc{\phi_n}{\cdot} \phi_n$  on the right hand side and bring the fractional power inside the sum we obtain
\begin{multline} \int_{\bbR}Q_{v+u}(\phi,\psi;I)^{2-s} \frac{\di u}{|u|^s} \\ \le \  \norm{P_I(H)\psi}^{2-2s}
 \sum_{n=1}^N\int_{I} \left ( \sum_{\kappa \in \sigma(K_v(E))} \ipc{\phi}{\Pi_\kappa(E)\phi}^{1-s} \abs{\ipc{\phi}{\Pi_\kappa(E)\phi_n}}^{s}\right )\abs{ \ipc{\phi_n}{  \left (H_v -E \right )^{-1} \psi}}^{s}  \di E. \end{multline}
The factor in parentheses on the right hand side is a correlator for $K_v(E)$, and thus bounded by $\norm{\phi}^{2-s}\norm{\phi_n}^s$.  The result follows.

\end{document}